\documentclass[preprint,12pt,authoryear]{elsarticle}

\usepackage{amssymb}
\usepackage{bm}
\usepackage{amsmath,amsfonts,amssymb,amscd,amsthm}
\usepackage{graphicx,xcolor,comment}
\usepackage{mathrsfs} 
\usepackage{multirow}
\usepackage{array}
\usepackage[colorlinks=true, urlcolor=black, linkcolor=black, citecolor=black]{hyperref}
\usepackage{multicol}
\usepackage{ragged2e}
\usepackage[english]{babel}
\usepackage{rotating}
\usepackage{enumerate}
\usepackage{booktabs}  
\usepackage{tikz}
\usepackage{bm}
\usepackage{animate}
\DeclareMathOperator{\interior}{int}
\usepackage{csquotes}
\usepackage{lineno}
\usepackage{float}
\usepackage[justification=justified,
format=plain]{caption}
\usepackage{subfig}
\usepackage{epstopdf}
\usepackage{wrapfig}
\usepackage{animate}
\usepackage{setspace}
\usepackage{enumitem}
\usepackage[top=2.5cm,bottom=2.5cm,left=2.5cm,right=2.5cm]{geometry}
\usepackage{cleveref}
\usepackage{thmtools}
\crefrangeformat{equation}{Eqs. (#3#1#4)--(#5#2#6)}
\crefrangeformat{figure}{Figs. (#3#1#4)--(#5#2#6)}
\crefrangeformat{table}{Tables. (#3#1#4)--(#5#2#6)}

\crefname{figure}{Fig.}{Figs.}
\crefname{table}{Table}{Tables}
\crefname{section}{Sec.}{Secs.}
\crefname{equation}{Eq.}{Eqs.}
\crefname{appendix}{Appendix}{Appendices}
\newcommand{\Sc}{\mathrm{Sc}}

\newtheorem{theorem}{Theorem}[section] 
\newtheorem{lemma}[theorem]{Lemma} 

\journal{a journal}

\begin{document}

\begin{frontmatter}



\title{Stochastic Suspended Sediment Dynamics in Semi-Bounded Open Channel Flows: A Reflected SDE Approach}


\author[inst1]{Manotosh Kumbhakar}

\author[inst1]{Christina W. Tsai\corref{cor1}}
\ead{cwstsai@ntu.edu.tw}
\cortext[cor1]{Corresponding author}

\affiliation[inst1]{organization={Department Of Civil Engineering},
	addressline={National Taiwan University}, 
	city={Taipei},
	postcode={10617}, 
	country={TAIWAN}}

\begin{abstract}
Stochastic processes, in the form of stochastic differential equations (SDEs), integrate stochastic elements to account for the inherent randomness in sediment particle trajectories in an open-channel turbulent flow. Accordingly, a stochastic diffusion particle tracking model (SDPTM) has been proposed in the literature to analyze suspended sediment dynamics. In this work, we develop a reflected stochastic diffusion particle tracking model (RSDPTM) for suspended sediment motion in a two-dimensional open channel flow based on reflected SDE, which is a mathematically consistent theory for stochastic processes in a bounded region. The Eulerian model given in terms of the Fokker-Planck equation (FPE) is also proposed by formulating boundary conditions for the confined domain. The existence and uniqueness of the solution to the proposed reflected SDE are proven, and the strong order of convergence of the projected Euler-Maruyama (EM) method is discussed. In order to correctly incorporate the physical mechanism of sediment-laden open-channel flow, an improved algorithm considering the threshold criteria of sediment suspension is proposed. The ensemble means, variances, and MSDs in both streamwise and vertical directions are discussed. It is observed that the particle motion in both directions follows anomalous diffusion, which is the deviation from normal or Fickian diffusion theory. Finally, the proposed model is validated through the suspended sediment concentration (SSC) distribution by comparing it with relevant experiential data, and the comparison shows an excellent agreement between the estimated and measured values of SSC. In summary, the proposed RSDPTM and improved algorithm may enhance our idea about the inherent randomness of suspended sediment motion in an open channel turbulent flow.   
\end{abstract}

%


\end{frontmatter}


\section{Introduction}
\label{sec:intro}
Sediment transport has been a subject of continuing investigation, gripping attention for over a century due to its fundamental significance in reservoir sedimentation control and management, river training networks, canal operation, pollutant transportation, etc. In general, sediment transport in open channel flow occurs in two distinct modes: bed load and suspended load. Bed load refers to the portion traveling close to the bed, while suspended load involves the part traveling in suspension above the bed \citep{dey2014fluvial}. These two modes of transport are separated by a thin fictitious line known as the bed-load layer thickness. Also, there are certain contributing factors to the threshold of bed load and suspended load mechanisms at which they start to occur. We focus on the suspension region of the flow field. 
\par 
The movement of suspended sediment particles in turbulent open channel flow can be conceptualized using Eulerian and Lagrangian methodologies. The Eulerian approach involves capturing the collective motion of particles by defining a control volume within the flow field and applying the mass conservation law to understand the temporal and spatial variations in volumetric suspended sediment concentration \citep{ancey2015stochastic,dey2014fluvial}. In contrast, the Lagrangian framework tracks individual particles within the flow field \citep{alsina2009measurements}. One stochastic variation of this approach utilizes stochastic differential equations (SDEs) driven by Brownian motion \citep{dimou1993random}. Brownian motion captures the erratic movement of sediment particles in turbulent flow and is mathematically expressed in SDEs through a random term representing the Wiener process. \cite{man2007stochastic} explored the dynamics of suspended sediment by formulating a model based on SDEs driven by this process, which is termed as the stochastic diffusion particle tracking model (SDPTM).  Later, this idea was extended in successfully analyzing various sediment-laden open channel flow scenarios, e.g., particle movements under extreme flow conditions \citep{oh2010stochastic,tsai2016incorporating}, interaction between flow and particles using a multivariate approach \citep{oh2018stochastic}, modeling of sediment motion under the effect of turbulent bursting phenomenon \citep{tsai2019modeling}, two-particle SDPTM incorporating particle correlation \citep{tsai2020stochastic}, incorporation of the memory effect of turbulent structures into the suspended sediment movements \cite{tsai2021incorporating},  identification of probable sedimentation sources by proposing a backward-forward SDPTM \citep{liu2021development}, characterization of sweep and ejection events and its implication to sediment transport \citep{wu2022probabilistic}, analysis of the influence of attached eddies in sediment transport \citep{huang2023modeling}, etc. 
\par 
The Wiener process exhibits specific characteristics, including a normally distributed Wiener increment with variance proportional to the time increment. This property, combined with increment stationarity, results in unbounded particle movement for larger time intervals. Consequently, the numerical simulation of stochastic differential equations (SDEs) driven by the Wiener process, such as the SDPTM, may yield impractical values like negatives or those beyond a specified domain. In the context of two-dimensional open-channel turbulent flow, where the domain is fully bounded vertically and semi-bounded in the streamwise direction, SDPTM modeling encounters challenges related to negative and undesirable values. Modifications are often implemented in numerical simulations to constrain values within the confined domain, including alternative methods like resampling the Wiener increment. However, these approaches may introduce biases, leading to inaccuracies in modeling treatments. To address these challenges, the mathematical community has long addressed these issues by refining SDEs using a concept called `local time,' resulting in the formulation of the reflected SDE (RSDE) \citep{singer2008partially}. In a general sense, local time refers to a stochastic process linked to another process, such as the Wiener process, quantifying the time a particle spends at a given boundary. RSDEs are also referred to as the Skorokhod problem, which initially established the existence and uniqueness of strong solutions to such SDEs \citep{skorokhod1961stochastic,skorokhod1962stochastic2}. Numerical solutions to SDEs require modifications of the deterministic methods due to the random term involved in the equation \citep{higham2001algorithmic}. Further, the numerical approaches for handling RSDEs are more challenging, requiring case-specific solutions. 
\par 
The SDPTM and related works mentioned above addressed the problem of bounded domain heuristically, needing a mathematically consistent formulation. Also, a detailed theoretical and numerical analysis for stochastic sediment transport modeling is yet to be carried out. Further, when sediment particles reach the bed, they are subject to the resuspension mechanism. Therefore, a robust algorithm for simulating particle trajectories may be needed to have a physically realistic model. Therefore, looking into the research gaps, the objectives of our work are to: 
\begin{enumerate}[label=\Roman*.]
	\item Formulate mathematically consistent Eulerian and Lagrangian models for suspended sediment dynamics in a semi-bounded open-channel turbulent flow. 
	\item Carry out the mathematical analysis of the proposed model, specifically, the existence and uniqueness of the solution to RSDE, analysis of the numerical method, and its order of convergence. 
	\item Propose an improved algorithm for particle trajectories incorporating the threshold of the sediment suspension mechanism.  
	\item Validate the proposed model with experimental data to check the efficiency of the modified algorithm. 
\end{enumerate}
These objectives are accomplished step by step in this study. In \cref{sec:math_model}, first, the existing SDPTM is discussed briefly. Then, both the Eulerian (FPE) and Lagrangian (RSDPTM) models incorporating the boundary effects are proposed, and the selection of hydraulic variables and parameters is discussed. In \cref{sec:existence}, the existence and uniqueness of the solution to general RSDES are discussed and then extended for the proposed RSDPTM. Numerical solution to the RSDPTM is explained in \cref{sec:numerical}. Next, in \cref{sec:result}, first, we consider some numerical tests to estimate the strong order of convergence of the projected EM. Then, the ensemble means, variances and MSDs of particle trajectories are discussed by proposing an improved algorithm. Also, the model is validated through the experimental data of sediment concentration distribution. Finally, conclusions and some possible future scopes are given in \cref{sec:conclusions}.      
\section{Mathematical Modeling}
\label{sec:math_model}
\subsection{Stochastic Diffusion Particle Tracking Model (SDPTM)}
\label{subsec:sdptm}
Let us consider a three-dimensional sediment-laden turbulence flow field. From a Lagrangian perspective, we model the sediment dynamics by looking into the trajectory of an individual particle. The Langevin equation can model the movement of particles in turbulent flow, which is a particular example of stochastic differential equations (SDEs) that combine the effect of deterministic and stochastic forces. \cite{man2007stochastic} developed the stochastic diffusion particle tracking model (SDPTM) for analyzing the motion of sediment particles in open channel turbulent flow. The generalized form of the equation representing the particle's location is given as follows:
\begin{equation}\label{eq1}
	d\boldsymbol{\Phi}_{t} = \bar{\boldsymbol{u}}(t,\boldsymbol{\Phi}_{t}) dt + \boldsymbol{\sigma} (t,\boldsymbol{\Phi}_{t}) d\boldsymbol{B}_{t}
\end{equation}
subject to an initial condition
\begin{equation}\label{eq2}
	\boldsymbol{\Phi}_{0} = \boldsymbol{\phi}_{0}
\end{equation}
where the stochastic process vector $\left\{\boldsymbol{\Phi}_{t}, t \in [0,\infty)\right\}$ is defined on a common probability space $\left(\mathbb{R},\mathcal{B}\left(\mathbb{R}\right),P\right)$ and a measurable space $\left(\mathbb{R},\mathcal{B}\left(\mathbb{R}\right)\right)$, in which $P$ denotes the probability measure. The initial values $\boldsymbol{\phi}_{0}$ are constants (non-random). Here, $\boldsymbol{\Phi}_{t} : \mathbb{R}^{3} \times [0,\infty) \to \mathbb{R}^{3}$ denotes the particle position vector, defined as
\begin{equation}\label{eq3}
	\boldsymbol{\Phi}_{t} = \begin{bmatrix}
		X_{t}\\
		Y_{t}\\
	    Z_{t}
	\end{bmatrix},
\end{equation} 
in the streamwise, transverse, and vertical directions, respectively. The term $\boldsymbol{B}_{t}$ represents the independent Brownian motion vector. In the context of sediment-laden turbulence, the mean drift velocity term $\bar{\boldsymbol{u}}(t,\boldsymbol{\Phi}_{t}):  \mathbb{R}^{3} \times [0,\infty) \to \mathbb{R}^{3}$ can be expressed in terms of mean velocity and diffusivity gradient, as follows:
\begin{equation}\label{eq4}
\bar{\boldsymbol{u}}(t,\boldsymbol{\Phi}_{t}) = {\setstretch{1.4}\begin{bmatrix}
		\bar{u}+\frac{\partial D_{x}}{\partial x}\\
		\bar{v}+\frac{\partial D_{y}}{\partial y}\\
		\bar{w}-w_{s}+\frac{\partial D_{z}}{\partial z}
	\end{bmatrix}},
\end{equation} 
where $\bar{u}$, $\bar{v}$, and $\bar{w}$ are mean fluid velocities in three directions; $D_{x}$, $D_{y}$, and $D_{z}$ are the sediment diffusivities along three directions; and $w_{s}$ is the sediment settling velocity (or, terminal fall velocity) that acts in the downward direction due to gravity. The diffusion coefficient tensor $\boldsymbol{\sigma} (t,\boldsymbol{\Phi}_{t}) : \mathbb{R}^{3} \times [0,\infty) \to \mathbb{R}^{3 \times 3}$ takes on the form:
\begin{equation}\label{eq5}
	\boldsymbol{\sigma} (t,\boldsymbol{\Phi}_{t}) = 
	\begin{bmatrix}
	    \sigma_{11} & \sigma_{12} & \sigma_{13}\\
	 	\sigma_{21} & \sigma_{22} & \sigma_{23}\\
		\sigma_{31} & \sigma_{32} & \sigma_{33}
	\end{bmatrix}
\end{equation} 
If the coordinate axes are aligned with the flow, then $\boldsymbol{\sigma} (t,\boldsymbol{\Phi}_{t})$ becomes a diagonal matrix, i.e., $\sigma_{ij}=0$ for $i \neq j$. The relationship between the diffusion coefficient $\boldsymbol{\sigma}$ and the sediment diffusivity $\boldsymbol{D}$ tensor can be given as \citep{man2007stochastic}:
\begin{equation}\label{eq6}
D_{ii} = \frac{1}{2} \left[\boldsymbol{\sigma \sigma^{T}}\right]_{i,i}
\end{equation} 
For isotropic turbulent flow, $D_{ii} = D_{x}, D_{y}$, and $D_{z}$ for $i=1,2$, and 3, respectively. Given the flow parameters and variables, such as $\boldsymbol{\bar{u}}$, $\boldsymbol{\sigma}$, and $\boldsymbol{\phi}_{0}$, one can simulate the governing SDE \cref{eq1} numerically. One such approach is the Euler-Maruyama scheme, which approximates the solution as follows:
\begin{equation}\label{eq7}
\boldsymbol{\Phi}_{t+\Delta t} =\boldsymbol{\Phi}_{t}+ \bar{\boldsymbol{u}}(t,\boldsymbol{\Phi}_{t}) \Delta t + \boldsymbol{\sigma} (t,\boldsymbol{\Phi}_{t}) \Delta \boldsymbol{B}_{t}
\end{equation} 
where $\Delta t$ is the time step. 
\par
The transport of sediments in open channel flow occurs mainly through two different modes, namely bed load and suspended load. In bed load transport, relatively heavier particles resting on the bed start rolling, sliding, or saltating (succession of small jumps) in the near-bed region of the flow. On the other hand, fine particles come into suspension in the main flow region due to turbulence, known as the suspended load. These two different phenomena are distinguished by a (hypothetical) thin line called the bed-load layer thickness \citep{dey2014fluvial}. Therefore, they are confined to a bounded region of the flow domain. \cite{man2007stochastic} developed SDPTM for modeling suspended sediment transport in open-channel turbulent flow. Later, their works were extended by incorporating several turbulent mechanisms as well as refined stochastic processes. However, none of the works considered a mathematically consistent formulation of a stochastic process in a bounded domain; instead, they handled the boundary by discarding the non-physical values. This is an important aspect while dealing with SDEs subject to boundary conditions, which are discussed below in detail. 
\par 
Due to the stationary increment of Brownian motion, the larger the interval, the larger the fluctuations on this interval. It means that the variance is proportional to time. This creates difficulty when modeling a physical phenomenon in a confined domain as the values produced by \cref{eq7} exceed the boundary of the domain and can become negative. Further, the square root term in the diffusion coefficient produces imaginary values in the numerical simulation of \cref{eq1}. These issues have been addressed using several ways in the literature \citep{fox1997stochastic,goldwyn2011stochastic,dangerfield2010stochastic,dangerfield2012boundary,lord2010comparison}. For example, the numerical simulation can be adjusted to have the positive solution or values within the bounded domain. Also, the stochastic term (Brownian increment) can be resampled to get the physical result. However, these approaches can still result in negative values or bias the outcome, as discussed in \cite{dangerfield2012modeling}. In the context of SDPTM, researchers have tackled this issue by forcing the numerical solution to be within the flow domain. 
\par 
The aforementioned issues were addressed long back in the mathematical formulation of SDEs subject to a bounded region \citep{skorokhod1961stochastic,skorokhod1962stochastic2}. The resulting equation is known as the reflected stochastic differential equation (RSDE). The concept of RSDEs has been applied successfully in many areas, such as to model the constrained animal movement \citep{brillinger2003simulating}, human metabolic process \citep{kawamura2006stochastic}, biochemical reaction kinetics \citep{niu2016modelling}, ion channel dynamics \citep{dangerfield2012modeling}, etc. In this work, we explore the RSDE for modeling suspended sediment movement in open channel turbulent flow. Both the Eulerian and Lagrangian approaches for modeling sediment transport subject to the boundary are discussed below in detail. 
\subsection{Reflected Stochastic Diffusion Particle Tracking Model (RSDPTM)}
\label{subsec:rsdptm}
We consider a two-dimensional open channel turbulent flow with uniform flow depth carrying sediment particles. For any streamwise distance ($0 \leq x < \infty$), the sediment particles come into suspension after a certain height, known as the reference level (say, $z=a$), and can make movements until the water surface ($z=h$). Therefore, $x \in [0,\infty)=\mathbb{R}^{+}$ and $z \in [a,h]=D_{1}$ (say), which implies the 2D domain is, say $D=\mathbb{R}^{+} \times D_{1}$. A schematic diagram is presented in \cref{fig_1}. The diagram in \cref{fig_3} illustrates both the wall effects and a comparison between standard Brownian motion without a wall and reflected Brownian motion near a wall. When we refer to a reflected Brownian-motion process, we are specifically describing a Brownian-motion particle that is reflected by a wall, with the reflection being momentary and involving energy loss, akin to a mirror reflection. While walls with partial-reflecting/partial-absorbing, totally absorbing, or delay-reflecting characteristics are intriguing and warrant investigation in the future, they are not within the scope of this study. In standard Brownian motion, a particle exhibits a mean displacement of 0 and a diffusion radius proportional to $dt$, where $dt$ represents the time increment. In contrast, the motion of a particle near a wall entails a mean displacement bias opposite to the wall, along with a suppressed diffusion. Next, based on the boundary, we formulate both Lagrangian and Eulerian equations for the suspended sediment dynamics. 

\begin{figure}[hbt!]
	\centering
	\includegraphics[height=4cm,width=10cm]{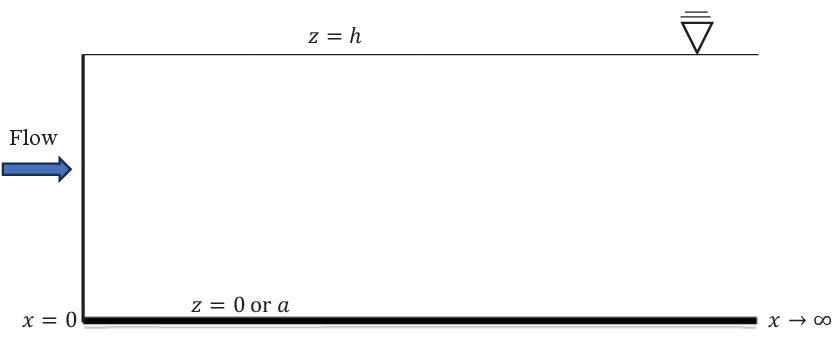}
	\caption{Schematic diagram of the flow domain.}
	\label{fig_1}
\end{figure}


\subsubsection{Eulerian Framework}
\label{subsub:euler}
In Eulerian approach, we focus on the collective behaviour of the particles rather than the individual one. To that end, here, we derive the Fokker-Planck equation for the sediment particles movement. In relation to \cref{eq1}, if $\boldsymbol{\Phi}_{t}$ is a function of say $Z_{t}$ only, then for any twice continuously differentiable function $g \in C^{2}(\mathbb{R})$ and It\^o process $Z_{t}$, we have from It\^o's lemma \citep{kloeden2011numerical}:
\begin{equation}\label{eq8}
	dg=\left[\frac{\partial g}{\partial t} + \bar{u}(z,t)\frac{\partial g}{\partial z} + \frac{1}{2} \sigma^{2}(z,t) \frac{\partial^{2} g}{\partial z^2}  \right] dt + \sigma(z,t) \frac{\partial g}{\partial z}dB_{t}
\end{equation}
Now taking the expectation on both sides of \cref{eq8} with the fact $\mathbb{E}(dB_{t}) = 0$, and then performing the integration by parts with $p(z,t)$,$\partial p(z,t)/\partial z$, $\partial^{2} p(z,t)/\partial z^{2}$ $\to 0$ as $z \to \mp \infty$, where $p(z,t)$ is the probability density function (PDF), we obtain:
\begin{equation}\label{eq9}
	\frac{\partial p}{\partial t}(z,t) = -\frac{\partial}{\partial z} \left[\left(\bar{w} - w_{s} + \frac{\partial D_{z}}{\partial z}\right) p(z,t)\right] + \frac{\partial^{2} }{\partial z^{2}} \left[D_{z} p(z,t)\right]
\end{equation}
\cref{eq9} is the one-dimensional Fokker-Planck equation of particle location PDF. Considering our model that is based on vertical and streamwise direction, the Fokker-Planck equation can be established using similar argument as follows:
\begin{align}\label{eq10}
	\frac{\partial p}{\partial t}(x,z,t) = -\frac{\partial}{\partial x} \left[\left(\bar{u} + \frac{\partial D_{x}}{\partial x}\right) p(x,z,t)\right] & -\frac{\partial}{\partial z} \left[\left(\bar{w} - w_{s} + \frac{\partial D_{z}}{\partial z}\right) p(x,z,t)\right] \nonumber \\ & + \frac{\partial^{2} }{\partial x^{2}} \left[D_{x} p(x,z,t)\right] + \frac{\partial^{2} }{\partial z^{2}} \left[D_{z} p(x,z,t)\right]
\end{align} 
The initial condition for the FPE \cref{eq10} corresponding to the condition \cref{eq2} can be given as:
\begin{align}\label{eq11}
p(x,z,0) = \delta (x-x_{0}) \delta (z-z_{0})
\end{align}
where $\delta$ denoted the Dirac delta function. 
\par 
Reflection mechanism at the boundaries can be invoked in \cref{eq10} by rewriting the equation and then applying some techniques. Let us first rewrite \cref{eq10} in the following form:
\begin{align}\label{eq12}
	\frac{\partial p}{\partial t}(\boldsymbol{s},t) + \sum_{j=1}^{2} \nabla \cdot J_{j}(\boldsymbol{s},t) = 0
\end{align}
where 
\begin{align}\label{eq13}
	J_{j}(\boldsymbol{s},t) = \bar{u}_{j}(\boldsymbol{s},t) p(\boldsymbol{s},t) - \sum_{i=1}^{2}  D_{ij}(\boldsymbol{s},t) p(\boldsymbol{s},t).
\end{align}
Here, $\boldsymbol{s} = (x,z)$, $\bar{\boldsymbol{u}} = (\bar{u}_{j})$, and $\boldsymbol{D} = [D_{ij}]$. \cref{eq13} represents the \textit{probability current}. Considering the region $D$ with volume $V$ and its boundary $\partial D$, one can have from the time derivative of total probability using \cref{eq12} and Gauss divergence theorem:
\begin{align}\label{eq14}
\frac{\partial}{\partial t} \int_{V} p(\boldsymbol{s},t) dV = \int_{V} \frac{\partial p(\boldsymbol{s},t)}{\partial t} dV =\sum_{j=1}^{2} \int_{V} \nabla \cdot J_{j}(\boldsymbol{s},t) dV = \sum_{j=1}^{2} \int_{\partial D} J_{j}(\boldsymbol{s},t) \cdot n_{D} dS
\end{align}
where $n_{D}$ is the outward drawn unit normal of $\partial D$. In order to make the total probability conserved in $D$, \cite{ferm2006conservative} deduced that $\boldsymbol{J} = 0$ for $\boldsymbol{s} \in \partial D$. This is the reflecting boundary condition associated with FPE. Also, for boundary at infinity, one can impose $\lim\limits_{s_{i} \to \infty} p(\boldsymbol{s},t) = 0$ \citep{gardiner1985handbook}. Thus, in relation with the domain described in \cref{subsec:rsdptm}, the boundary conditions become:
\begin{equation}\label{eq15}
\begin{aligned}
 - \left(\bar{w} - w_{s} + \frac{\partial D_{z}}{\partial z}\right) p(x,z,t) + \frac{\partial}{\partial z} \left[D_{z} p(x,z,t)\right]& = 0 ~~\text{at}~~ z=a~~\text{and}~~z=h\\
-\left(\bar{u} + \frac{\partial D_{x}}{\partial x}\right) p(x,z,t) + \frac{\partial}{\partial x} \left[D_{x} p(x,z,t)\right] &= 0~~ \text{at}~~ x=0\\
\lim\limits_{x \to \infty} p(x,z,t) &= 0 
\end{aligned}
\end{equation}
The FPE \cref{eq10} can be solved together with the initial condition \cref{eq11} and boundary conditions \cref{eq15}. 
\subsubsection{Lagrangian Framework}
\label{subsub:lagrange}
In order to apply the boundary conditions to the SD-PTM \cref{eq1}, i.e., to keep the solution \cref{eq7} inside the domain $D$, we need to decompose the stochastic process $\boldsymbol{\Phi}_{t}$ as a sum of two stochastic processes, say $\boldsymbol{\Phi}_{t} = \boldsymbol{\Psi}_{t} + \boldsymbol{K}_{t}$. Here, $\boldsymbol{\Psi}_{t}$ is governed by the original SD-PTM \cref{eq1}, i.e., it describes the behaviour of the process $\boldsymbol{\Phi}_{t}$ in the interior of the domain $D$, for which we must have $\boldsymbol{\Phi}_{t=0} = \boldsymbol{\Psi}_{t=0}$ and $\boldsymbol{\Phi}_{t} = \boldsymbol{\Psi}_{t}$ for $\boldsymbol{\Psi}_{t} \in \interior (D)$. The second process $\boldsymbol{K}_{t}$ determines the behaviour at the boundary and reflects the solution into $D$. Its initial value is set as $\boldsymbol{K}_{t} = 0$. This process $\boldsymbol{K}_{t}$ may also be thought as the minimal process, which forces $\boldsymbol{\Phi}_{t}$ to remain in the domain $D$. The measure induced by this process must be concentrated at those times, say $t_{e}$, where $\boldsymbol{\Phi}_{t} \in \partial D$. Mathematically, one can write:
\begin{equation}\label{eq16}
\left|\boldsymbol{K}\right|_{t} = \int_{0}^{t} \boldsymbol{1}_{\left\{\boldsymbol{\Phi}_{t} \in \partial D\right\}} d\left|\boldsymbol{K}\right|_{s}
\end{equation}
where $\boldsymbol{1}_{\left\{\boldsymbol{\Phi}_{t} \in \partial D\right\}}$ is the indicator function. The authors in \cite{bayer2010adaptive} informally call the process $\boldsymbol{K}_{t}$ as the \textit{local time}. This is because the measure induced by $\boldsymbol{K}_{t}$ characterizes the amount of time $\boldsymbol{\Phi}_{t}$ spends at the boundary. \cref{eq16} discusses about the behaviour at the boundary; however, we are yet to specify how the reflection at the boundary will happen. To that end, we assume that the process $\boldsymbol{K}_{t}$ will reflect $\boldsymbol{\Phi}_{t}$ into $\interior (D)$ in the direction of an inward drawn unit normal. For example, at the surface, i.e., $Z_{t} = h$, $\boldsymbol{K}_{t}$ will force the process in the downward direction, and at the bottom boundary, i.e., $Z_{t} = a$, it will push the process in the upward direction. Considering these, we can write:
\begin{equation}\label{eq17}
	\boldsymbol{K}_{t} = \int_{0}^{t} \gamma(s) d\left|\boldsymbol{K}\right|_{s}
\end{equation}
where $\gamma (s) \in \mathcal{N}(\boldsymbol{\Phi}_{s})$ if $\boldsymbol{\Phi}_{s} \in \partial D$ and $\mathcal{N}(\boldsymbol{x})$ denotes the set of all inward drawn unit normal vectors at the point $\boldsymbol{x}$. Therefore, reflected stochastic diffusion particle tracking model (RSD-PTM) reads as:
\begin{equation}\label{eq18}
d\boldsymbol{\Phi}_{t} = \bar{\boldsymbol{u}}(t,\boldsymbol{\Phi}_{t}) dt + \boldsymbol{\sigma} (t,\boldsymbol{\Phi}_{t}) d\boldsymbol{B}_{t} + d\boldsymbol{K}_{t}
\end{equation} 
\subsubsection{Hydraulic Variables and Parameters}
\label{subsub:modelparameters}
In the context of suspended sediment transport, the drift and diffusion coefficient in \cref{eq18} contain several hydraulic variables and parameters, namely the mean flow velocities, settling velocity, sediment diffusivity, as can be seen from \cref{eq4,eq5,eq6}. Based on our consideration of a two-dimensional flow, the mean velocities $\bar{u}$ and $\bar{w}$ should be prescribed. Generally speaking, the mean velocity profiles in turbulent flow are obtained using Reynolds-averaged Navier-Stokes (RANS) equation. The classical logarithmic law of velocity reads as follows \citep{dey2014fluvial}:
\begin{equation}\label{eq19}
	\bar{u} = \frac{u_{*}}{\kappa}\ln \frac{z}{z_{0}}  
\end{equation}
where $u_{*}$ is the shear velocity, $\kappa$ is the von-K\'{a}rm\'{a}n constant, and $z_{0}$ is the zero-velocity level. The constant $\kappa$ is typically assumed as 0.41. Introdcuing the concept of equivalent roughness $k_{s}$, \cite{nikuradse1933stromungsgesetze} divided the flow regimes into smooth, rough, and transitional. The shear Reynolds number, defined as $R_{*} = u_{*}k_{s}/\nu_{f}$, classifies the flow regimes:hydraulically smooth flow ($R_{*} \leq 5$), hydraulically rough flow ($R_{*} \geq 70$), and hydraulically transitional flow ($5 < R_{*} < 70$). Accordingly, based on Nikuradse's pipe flow experiment, the zero-velocity level is given as:
\begin{equation}\label{eq19_zero_vel_level}
	z_{0} = \begin{cases}
		0.11\frac{\nu_{f}}{u_{*}} & \text{if } R_{*} \leq 5\\
		\frac{k_{s}}{30} & \text{if } R_{*} \geq 70\\
		0.11\frac{\nu_{f}}{u_{*}} + \frac{k_{s}}{30} & \text{if } 5<R_{*}<70
	\end{cases}
\end{equation}
We choose the following formula for $k_{s}$ proposed by \cite{sumer1996velocity}:
\begin{equation}\label{eq19_roughness}
	k_{s} = \begin{cases}
		d_{50}\left[2 + 0.6\Theta^{2.5}\right] & \text{if } w_{s} > 0.9u_{*}\\
		d_{50}\left[4.5+0.25\Theta^{2.5} \exp \left(0.6W_{*}^{4}\right) \right] & \text{if } w_{s} \leq 0.9u_{*}
	\end{cases}
\end{equation}
where $\Theta$ is the Shields parameter, defined as $\Theta = u_{*}^{2}/\left(\Delta g d_{50}\right)$, in which $\Delta = s-1$, $s$ being the relative density of sediment; $w_{s}$ is the settling velocity; and $W_{*}$ is defined as $W_{*} = w_{s}/\sqrt{\Delta g d_{50}}$.
For the flow configuration considered in this study, the mean vertical velocity $\bar{w}$ is zero. 
\par 
The sediment diffusion coefficient in the vertical direction, $D_{z}$, is proportional to the turbulent diffusivity $D_{0z}$. The proportional coefficient is known as the turbulent Schmidt number, defined as $Sc=D_{z}/D_{0z}$ \citep{dey2014fluvial}. The turbulent diffusivity or eddy viscosity is mainly derived
using the analogy of Newton’s law of viscosity for turbulent flow. One of the widely used eddy viscosity profile is given by the following parabolic equation \citep{rijn1984sediment,graf2002suspension}:
\begin{equation}\label{eq20}
	D_{0z} = \kappa u_{*} z\left(1-\frac{z}{h}\right)
\end{equation}
The parameter $Sc$ is an important quantity to accurately predict the sediment concentration distribution \citep{graf2002suspension}. There are several formulae available in the literature for estimating $Sc$. We consider the most recent formula proposed by \cite{pal2016effect}, and is given for dilute sediment-water mixture flow as follows:
\begin{equation}\label{eq21}
	\Sc = 0.033 \left(\frac{w_{s}}{u_{*}}\right)^{0.931} \xi_{a}^{-1.196} c_{a}^{-0.118}
\end{equation}
Here, $\xi_{a}$ is the normalized reference level from where the suspension region starts, and $c_{a}$ is the suspended sediment concentration measure at $\xi_{a}$. The discussion on sediment diffusivity in the streamwise direction can be approached by considering \cite{socolofsky2005special}'s proposal. According to their findings, the longitudinal and transverse components are deemed equivalent, with no observed boundary effects in these two directions. The transverse component of sediment diffusivity, $D_{y}$, can be taken from \cite{fischer1979mixing}, and hence we have the following equation for the streamwise component:
\begin{equation}\label{eq22}
	D_{x} \approx D_{y} = 0.15 u_{*} z
\end{equation} 
The settling velocity $w_{s}$, also known as the terminal fall velocity, is the constant velocity attained by a particle when moving down through water column. Generally, $w_{s}$ depends on the size of the particle through particle diameter. We estimate $w_{s}$ from the formula given by \cite{cheng1997simplified}, as follows: 
\begin{equation}\label{eq23}
	w_{s}  = \frac{\nu_{f}}{d} \left(\sqrt{25+1.2d_{*}^{2}}-5\right)^{3/2}
\end{equation}
where $d_{*}=\left(\Delta g/\nu_{f}^{2}\right)^{1/3}d$. Here, $d$ denotes the particle diameter, $g$ is the acceleration due to gravity, $\nu_{f}$ is the kinematic viscosity of fluid, and $\Delta$ is the submerged specific gravity. Using these expressions \cref{eq19,eq20,eq21,eq22,eq23}, we can reformulate the 2D Langevin model in the following form:
\begin{equation}\label{eq24}
\begin{bmatrix}
	dX_{t}\\
	dZ_{t}
\end{bmatrix} = \begin{bmatrix}
\frac{u_{*}}{\kappa}\ln \frac{Z_{t}}{z_{0}}\\
-w_{s}+\kappa u_{*} \Sc \left(1-\frac{2Z_{t}}{h}\right)
\end{bmatrix} dt +\begin{bmatrix}
\sqrt{0.30 u_{*} Z_{t}} & 0\\
0 & \sqrt{2 \kappa u_{*} \Sc Z_{t}\left(1-\frac{Z_{t}}{h}\right)}
\end{bmatrix} \begin{bmatrix}
dB_{1t}\\
dB_{2t}
\end{bmatrix} + \begin{bmatrix}
dK_{1t}\\
dK_{2t}
\end{bmatrix}
\end{equation} 

\section{Existence and Uniqueness of the Solution to RSDPTM}
\label{sec:existence}
\cite{tanaka1979stochastic} provided the proof for the existence and uniqueness of solutions to the RSDEs under the condition that the solution domain is convex, and specific constraints are met by the drift and diffusion coefficients. Let us rewrite the proposed RSDE. 
\par 
Let $D \in \mathbb{R}^{2}$ and $\boldsymbol{B}_{t}$ denotes two independent Wiener processes. Given an $\mathbb{R}^{2}$-valued function $\bar{\boldsymbol{u}}(t,\boldsymbol{\Phi}_{t})$ and $\mathbb{R}^{2} \times \mathbb{R}^{2}$-valued function $\boldsymbol{\sigma} (t,\boldsymbol{\Phi}_{t})$, both being defined on $\mathbb{R}^{+} \times \bar{D}$. We consider the following RSDE: 
\begin{equation}\label{eq25}
	\left\{\begin{array}{l}	d\boldsymbol{\Phi}_{t} = \bar{\boldsymbol{u}}(t,\boldsymbol{\Phi}_{t}) dt + \boldsymbol{\sigma} (t,\boldsymbol{\Phi}_{t}) d\boldsymbol{B}_{t} + d\boldsymbol{K}_{t}\\ 	\boldsymbol{\Phi}_{0} = \boldsymbol{\phi}_{0}\end{array}\right.
\end{equation}
where the solution $\boldsymbol{\Phi}_{t}=\left(X_{t}, Z_{t}\right) \in \bar{D}$. Assume that $\bar{\boldsymbol{u}}(t,\boldsymbol{\Phi}_{t})$ and $\boldsymbol{\sigma} (t,\boldsymbol{\Phi}_{t})$ are Borel measurable. Considering $D$ as a convex domain in $\mathbb{R}^{2}$, \cite{tanaka1979stochastic} provided the following theorems. 
\begin{theorem}\label{thm1}
	\citep{tanaka1979stochastic}. If $\bar{\boldsymbol{u}}(t,\boldsymbol{\Phi}_{t})$ and $\boldsymbol{\sigma} (t,\boldsymbol{\Phi}_{t})$ are bounded continous on $\mathbb{R}^{+} \times \bar{D}$, then on some probability space $\left(\Omega, \mathcal{F}, P\right)$, we can find a two-dimensional Brownian motion $\boldsymbol{B}_{t}$ in such a way that \cref{eq25} has a solution. 
\end{theorem}
\begin{theorem}\label{thm2}
	\citep{tanaka1979stochastic}. Let $\rho$ and $\bar{\rho}$ satisfy
	\begin{equation}\label{eq26}
		\int_{0^{+}} \left[\rho^{2}(u) u^{-1} + \bar{\rho}(u)\right]^{-1} du = \infty,
	\end{equation}
	\begin{equation}\label{eq27}
		\rho^{2}(u) u^{-1} + \bar{\rho}(u)~\text{is concave}
	\end{equation}
	Then, for any $\bar{\boldsymbol{u}}(t,\boldsymbol{\Phi}_{t})$ and $\boldsymbol{\sigma} (t,\boldsymbol{\Phi}_{t})$ satisfying $$\|\bar{\boldsymbol{u}}(t,\boldsymbol{\Phi}_{t})-\bar{\boldsymbol{u}}(t,\boldsymbol{\Psi}_{t})\| \leq \bar{\rho} \|\boldsymbol{\Phi}_{t}-\boldsymbol{\Psi}_{t}\|,~ \|\boldsymbol{\sigma} (t,\boldsymbol{\Phi}_{t})-\boldsymbol{\sigma} (t,\boldsymbol{\Psi}_{t})\| \leq \rho \|\boldsymbol{\Phi}_{t}-\boldsymbol{\Psi}_{t}\|,$$ the pathwise uniqueness of solutions holds for \cref{eq25}.
\end{theorem}
For our RSDE given by \cref{eq24}, we need to prove each of the conditions to ensure the existence and uniqueness of solutions. Each of the steps is given in what follows.
\begin{lemma}\label{lemma1}
	The solution domain $D=\mathbb{R}^{+} \times D_{1}$ given in \cref{subsec:rsdptm} is convex.
\end{lemma}
\begin{proof}
	The domain in our model is given as $D = \left(0,\infty\right) \times \left(a,h\right)$. To prove that $D$ is convex, we need to show that for two points $\left(x_{1},y_{1}\right)$ and $\left(x_{2},y_{2}\right)$ in $D$, the linear combination $\lambda \left(x_{1},y_{1}\right) + (1-\lambda)\left(x_{2},y_{2}\right) \in D$, where $0 \leq \lambda \leq 1$. Considering the linear combination, we have the coordinate $\left(\lambda x_{1}+(1-\lambda)x_{2}, \lambda y_{1}+(1-\lambda) y_{2}\right)$. Since $0 < x_{1}, x_{2} < \infty$ and $a < y_{1}, y_{2} < h$, we have $$ 0 < \lambda x_{1}+(1-\lambda)x_{2} < \infty~\text{and}~a<\lambda y_{1}+(1-\lambda)y_{2}<h$$
	Hence, $\left(\lambda x_{1}+(1-\lambda)x_{2}, \lambda y_{1}+(1-\lambda) y_{2}\right) \in D$. This completes the proof that $D$ is a convex set. 
\end{proof}
Before establishing the existence and uniqueness of the solution of RSDPTM, we prove the following lemma that will be useful. 
\begin{lemma}\label{lemma2}
	If $f(x)$ is bounded Lipschitz function on a domain $\Omega$, then it is H\"{o}lder continuous on $\Omega$ with exponent $\alpha$, where $\alpha \in (0,1)$.  
\end{lemma}
\begin{proof}
	For $\|x-y\| \leq 1$, we have $$ \|f(x)-f(y)\| \leq L \|x-y\| =L \|x-y\|^{\alpha} \|x-y\|^{1-\alpha} \leq L \|x-y\|^{\alpha}$$
	Since $f(x)$ is bounded, we have $\|f(x)\| \leq C$. Therefore, for $\|x-y\| > 1$, $$\|f(x)-f(y)\| \leq \|f(x)\|+\|f(y)\| \leq 2C \leq 2C \|x-y\|^{\alpha}$$
	Hence, $f(x)$ is H\"{o}lder continuous on $\Omega$ with exponent $\alpha$.
\end{proof}
Our case deals with the finite dimensional vector spaces over the reals. Therefore, all norms are equivalent. For convenience, we choose maximum norm, which is defined for a vector $x=\left(x_{1},x_{2},...,x_{n}\right)$, as $\|x\|_{\infty} = \max \left(|x_{1}|,...,|x_{n}|\right)$.
\begin{lemma}\label{lemma3}
	The RSDPTM given by \cref{eq24} satisfies the conditions in \cref{thm1,thm2}, which ensures the existence and uniqueness of its solutions.
\end{lemma}
\begin{proof}
	The drift and diffusion coefficients for \cref{eq24} are $\bar{u}_{1}(\boldsymbol{\Phi}_{t}) = \frac{u_{*}}{\kappa}\ln \frac{Z_{t}}{z_{0}}$, $\bar{u}_{2}(\boldsymbol{\Phi}_{t}) = -w_{s}+\kappa u_{*} \Sc \left(1-\frac{2Z_{t}}{h}\right)$, $\sigma_{11}(\boldsymbol{\Phi}_{t}) = \sqrt{0.30u_{*}Z_{t}}$, and $\sigma_{22} (\boldsymbol{\Phi}_{t})= \sqrt{2\kappa u_{*}\Sc Z_{t}\left(1-\frac{2Z_{t}}{h}\right)}$. Clearly, all these functions are bounded-continuous on the domain $\bar{D}$. Therefore, according to \cref{thm1}, the RSDPTM \cref{eq24} has a solution. We now need to verify \cref{thm2} in order to establish uniqueness of the solution. For convenience, considering the constants, we may rewrite the coefficients as $\bar{u}_{1} = p_{1}\ln Z_{t}+p_{2}$, $\bar{u}_{2} = p_{3}-p_{4}Z_{t}$, $\sigma_{11} = p_{5} \sqrt{Z_{t}}$, and $\sigma_{22}= p_{6}\sqrt{Z_{t}\left(1-p_{7}Z_{t}\right)}$. Let us consider $\boldsymbol{\Phi}_{1,t} = \left(X_{1,t}, Z_{1,t}\right)$ and $\boldsymbol{\Phi}_{2,t} = \left(X_{2,t}, Z_{2,t}\right)$. Then, 
	\begin{align}\label{eq28}
		\|\bar{u}_{1} \left(\boldsymbol{\Phi}_{1,t}\right) - \bar{u}_{1}\left(\boldsymbol{\Phi}_{2,t}\right)\| & = |p_{1}| \left\|\ln \frac{Z_{1,t}}{Z_{2,t}}\right\| = |p_{1}| \left\|\ln \left(1+\left(\frac{Z_{1,t}}{Z_{2,t}}-1\right)\right)\right\| \leq |p_{1}| \left\|\left(\frac{Z_{1,t}}{Z_{2,t}}-1\right)\right\| \nonumber \\ & = |p_{1}| \frac{\|Z_{1,t}-Z_{2,t}\|}{\|Z_{2,t}\|} \leq \frac{|p_{1}|}{|a|} \|Z_{1,t}-Z_{2,t}\| \nonumber \\ & \leq L \max \left(\|X_{1,t}-X_{2,t}\|,\|Z_{1,t}-Z_{2,t}\|\right) = L \|\boldsymbol{\Phi}_{1,t}-\boldsymbol{\Phi}_{2,t}\|_{\infty}
	\end{align}
	We used the facts that $\ln \left(1+\bullet\right) \leq \bullet~\text{for}~\bullet >-1$ and $a \leq Z_{t} \leq h$. Now, 
\begin{align}\label{eq29}
	\|\bar{u}_{2} \left(\boldsymbol{\Phi}_{1,t}\right) - \bar{u}_{2}\left(\boldsymbol{\Phi}_{2,t}\right)\| = |p_{4}| \|Z_{1,t}-Z_{2,t}\| &\leq L \max \left(\|X_{1,t}-X_{2,t}\|,\|Z_{1,t}-Z_{2,t}\|\right) \nonumber \\ & = L \|\boldsymbol{\Phi}_{1,t}-\boldsymbol{\Phi}_{2,t}\|_{\infty}
\end{align}
It may be noted that we use $L$ as a generic positive constant throughout the paper. This constant may vary from case to case but is independent of $\boldsymbol{\Phi}_{1,t}$ and $\boldsymbol{\Phi}_{2,t}$. In the above, we showed that both $\bar{u}_{1}$ and $\bar{u}_{2}$ are Lipschitz continuous. Since the sum of Lipschitz continuous functions is also Lipschitz continuous, $\boldsymbol{\bar{u}}$ is also Lipschitz continuous and we have
\begin{equation}\label{eq30}
	\|\bar{\boldsymbol{u}}(t,\boldsymbol{\Phi}_{1,t}) - \bar{\boldsymbol{u}}(t,\boldsymbol{\Phi}_{2,t})\| \leq L \|\boldsymbol{\Phi}_{1,t}-\boldsymbol{\Phi}_{2,t}\|_{\infty}
\end{equation}
Now
\begin{align}\label{eq31}
	\left\|\sigma_{11}(\boldsymbol{\Phi}_{1,t}) - \sigma_{11}(\boldsymbol{\Phi}_{2,t})\right\| &=|p_{5}| \left\|\sqrt{Z_{1,t}}-\sqrt{Z_{2,t}}\right\| =|p_{5}| \frac{\|Z_{1,t}-Z_{2,t}\|}{\sqrt{Z_{1,t}} +\sqrt{Z_{2,t}}} \nonumber \\ & = |p_{5}| \sqrt{\left\|Z_{1,t}-Z_{2,t}\right\|} \frac{\sqrt{\|Z_{1,t}-Z_{2,t}\|}}{\sqrt{Z_{1,t}} +\sqrt{Z_{2,t}}} \leq |p_{5}|\sqrt{\left\|Z_{1,t}-Z_{2,t}\right\|} \nonumber \\ & \leq L \max \left(\|X_{1,t}-X_{2,t}\|^{1/2}, \|Z_{1,t}-Z_{2,t}\|^{1/2}\right) = L \left(\|\boldsymbol{\Phi}_{1,t}-\boldsymbol{\Phi}_{2,t}\|_{\infty}\right)^{1/2}
\end{align}
\begin{align}\label{eq32}
	\left\|\sigma_{22}(\boldsymbol{\Phi}_{1,t}) - \sigma_{22}(\boldsymbol{\Phi}_{2,t})\right\| &=  \left\|\sqrt{Z_{1,t}\left(1-p_{7}Z_{1,t}\right)}-\sqrt{Z_{2,t}\left(1-p_{7}Z_{2,t}\right)}\right\| \nonumber \\ & = \frac{\|Z_{1,t}\left(1-p_{7}Z_{1,t}\right)-Z_{2,t}\left(1-p_{7}Z_{2,t}\right)\|}{\sqrt{Z_{1,t}\left(1-p_{7}Z_{1,t}\right)}+\sqrt{Z_{2,t}\left(1-p_{7}Z_{2,t}\right)}} \nonumber \\ & = \frac{\|\left(Z_{1,t}-Z_{2,t}\right) \left[1-p_{7}\left(Z_{1,t}+Z_{2,t}\right)\right]\|}{\sqrt{Z_{1,t}\left(1-p_{7}Z_{1,t}\right)}+\sqrt{Z_{2,t}\left(1-p_{7}Z_{2,t}\right)}} \leq L\left\|Z_{1,t}-Z_{2,t}\right\| \nonumber \\ & \leq L \max \left(\|X_{1,t}-X_{2,t}\|, \|Z_{1,t}-Z_{2,t}\|\right) = L \|\boldsymbol{\Phi}_{1,t}-\boldsymbol{\Phi}_{2,t}\|_{\infty}
\end{align}
Since $\sigma_{22}(\boldsymbol{\Phi}_{1,t})$ is a bounded Lipchitz function on the domain $\bar{D}$, according to \cref{lemma2}, it is 1/2- H\"{o}lder continuous on the domain. Thus, both $\sigma_{11}(\boldsymbol{\Phi}_{1,t})$ and $\sigma_{22}(\boldsymbol{\Phi}_{1,t})$ are 1/2- H\"{o}lder continuous functions. Since the sum of  H\"{o}lder continuous functions is  H\"{o}lder continuous, $\boldsymbol{\sigma}$ is also  H\"{o}lder continuous, and we have
\begin{equation}\label{eq33}
	\|\boldsymbol{\sigma}(t,\boldsymbol{\Phi}_{1,t}) - \boldsymbol{\sigma}(t,\boldsymbol{\Phi}_{2,t})\| \leq L \left(\|\boldsymbol{\Phi}_{1,t}-\boldsymbol{\Phi}_{2,t}\|_{\infty}\right)^{1/2}
\end{equation}
Following \cref{thm2}, from \cref{eq30,eq33}, we have $\rho \left(u\right) = L u^{1/2}$ and $\bar{\rho} \left(u\right) = L u$. Therefore, $		\rho^{2}(u) u^{-1} + \bar{\rho}(u) = L^2+Lu$, which is concave. Further,
\begin{equation}\label{eq34}
	\int_{0^{+}} \left[\rho^{2}(u) u^{-1} + \bar{\rho}(u)\right]^{-1} du = \int_{0^{+}} \frac{du}{L^{2}+Lu} = \infty
\end{equation}
Hence, according to \cref{thm2}, the RSDPTM model \cref{eq24} has a unique solution. 
\end{proof}
\section{Numerical Solution of RSDPTM}
\label{sec:numerical}
In the previous section, we found that the RSDPTM \cref{eq24} has a solution, and it is unique. Since \cref{eq24} is not analytically tractable, we need to approximate it numerically. There are several numerical schemes available in the literature for simulating RSDEs. However, they are limited as compared to ordinary SDEs. The numerical methods for SDEs can be broadly categorized into penalization and projection methods. The penalization method involves constructing solutions to RSDEs by approximating diffusion processes, wherein the reflecting process is substituted with a penalty term, $\beta_{\lambda}(y)$. Convergence to the solution of the RSDE is achieved as $\lambda \downarrow 0$. Several authors have explored this method in relation to its strong and weak order of convergence under different conditions \citep{menaldi1983stochastic,laukajtys2003penalization,liu1995discretization,ding2008splitting}. However, a drawback of penalization methods is the potential for numerical solutions to exit the domain $D$, even when the exact solution to the RSDE does not. Due to this limitation, such methods are not employed in our approach. On the other hand, the projection method approximates the solution to the SDE without reflection. If the numerical solution ventures outside the domain $D$, it is projected back onto the domain. These methods ensure that numerical solutions stay within the desired region \citep{chitashvili1981strong,saisho1987stochastic,lepingle1995euler,slominski1994approximation,slominski1995some,liu1995discretization,pettersson1995approximations}. Therefore, we present an algorithm for obtaining numerical solution to the RSDE utilizing the projection method.
\par 
The projection method is a straightforward extension of the Euler-Maruyama (EM) method described in \cref{subsec:sdptm}. In this method, first, we compute the non-reflected process at time $t+\Delta t$ using the EM algorithm (see \cref{eq7}). If the resulting value falls within the closure $\bar{D}$ of $D$, then the process at the next time step is set to this value. But, if it lies outside $\bar{D}$, then the process is adjusted to be the orthogonal projection of this point onto the boundary of $D$. The orthogonal projection map onto $D$, denoted by $\Pi$, is explained below. Let $\Delta t$ denote the fixed time step. The projection algorithm is as follows:
\begin{itemize}[left=34pt]
	\item[\textbf{Step I:}] Set $t=0$. Input drift and diffusion coefficients $\bar{\boldsymbol{u}}$ and $\boldsymbol{\sigma}$, initial condition $\boldsymbol{\phi}_{0}$, and time step $\Delta t$. \vspace{0.1cm} 
	\item[\textbf{Step II:}] Generate Brownian increments $\boldsymbol{\Delta B}_{p,t}$ for $p=1,2$ as independent Gaussian random variables with mean 0 and variance $\Delta t$. Then, simulate
	\begin{equation}\label{eq35}
		\bar{\boldsymbol{\Phi}}_{t+\Delta t} =\boldsymbol{\Phi}_{t}+ \bar{\boldsymbol{u}}(t,\boldsymbol{\Phi}_{t}) \Delta t + \boldsymbol{\sigma} (t,\boldsymbol{\Phi}_{t}) \Delta \boldsymbol{B}_{t}
	\end{equation}
	at time $t+\Delta t$.\vspace{0.1cm}
	\item[\textbf{Step III:}] If $\bar{\boldsymbol{\Phi}}_{t+\Delta t} \in \bar{D}$ then set $\boldsymbol{\Phi}_{t+\Delta t} = \bar{\boldsymbol{\Phi}}_{t+\Delta t}$. Otherwise, take orthogonal projection $\boldsymbol{\Phi}_{t+\Delta t} = \Pi \left(\bar{\boldsymbol{\Phi}}_{t+\Delta t}\right)$.\vspace{0.1cm}
	\item [\textbf{Step IV:}] Set $t = t+ \Delta t$ and return to step II. 
\end{itemize}
\cref{fig_3} shows how the orthogonal projection occurs in modeling suspended sediment dynamics. It can be noted that the projection operator $\Pi$, for general domains, constructs a minimization problem for finding the projection of a point \citep{dangerfield2012modeling}. However, for the proposed RSDPTM, the projection of a point is the mirror reflection with respect to the boundary. For example, if a point $z=a-\epsilon$ lies outside the bottom boundary $z=a$, then its projection will be at $z=a+\epsilon$. 
\begin{figure}[hbt!]
	\centering
	\includegraphics[scale=0.6]{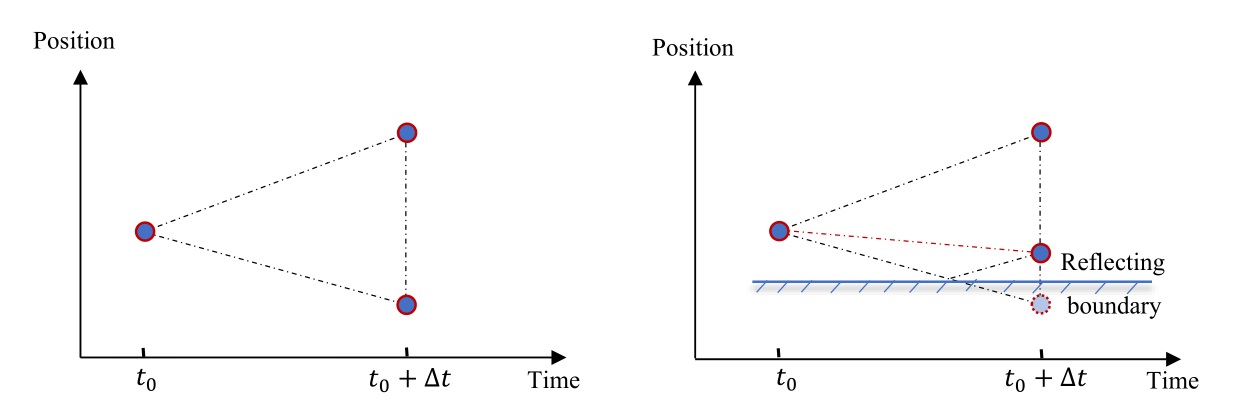}
	\caption{A sketch of the standard SDE and RSDE for particle trajectory. The wall-drift and suppression effects can be recognized in the right-side figure. Here, $t_{0}$ is an arbitrary instant of time, and $\Delta t$ is the time increment.}
	\label{fig_3}
\end{figure}
Now, it is important to mention the convergence of the projected EM method. In the context of SDEs, the notion of strong error is often useful. The order of convergence determines the rate at which a numerical method converges. The existing theoretical convergence of the projection method is summarized below. 
\begin{theorem}\label{thm3}
	\citep{liu1995discretization}.  Suppose that the coefficients $\boldsymbol{\bar{u}}$ and $\boldsymbol{\sigma}$ are Lipschitz continuous and there exists a constant $C_{K}$ such that 
	\begin{equation}\label{eq36}
		\left|\bar{u}_{i}\right| \leq C_{K} \left|\Phi_{i}\right|~~\text{and}~~\left|\sigma_{ii}\right| \leq C_{K} \left|\Phi_{i}\right|~~\forall~~ \Phi \in \bar{D}
	\end{equation}
	Then $\Phi_{t}^{\Delta t}$ defined by the projection scheme discussed in \cref{eq35} converges to $\Phi_{t}$ on $\left[0,T\right]$ in the mean square sense with order one, i.e., 
	\begin{equation}\label{eq37}
		\mathbb{E} \left[\left|\Phi_{T}-\Phi_{T}^{\Delta t}\right|^{2}\right] = \mathcal{O}\left(\Delta t\right)
	\end{equation}
\end{theorem}
\par 
For general coefficients $\bar{u}$ and $\sigma$, \cite{pettersson1995approximations} proved the following result that if $\bar{u}$ and $\sigma$ are bounded and Lipschitz continuous functions, then
\begin{equation}\label{eq38}
	\mathbb{E} \left[\sup_{0 \leq t \leq T}\left|\Phi_{t}-\Phi_{t}^{\Delta t}\right|^{2}\right] = \mathcal{O}\left(\Delta t \log \left(1/\Delta t\right)\right)
\end{equation}
holds for small $\Delta t$. 
\section{Results and Discussion}
\label{sec:result}
Here, we first verify the order of convergence of the projection method numerically. Then, the Eulerian (FPE) and Lagrangian (RSDPTM) approaches are compared to check the consistency. Next, the ensemble quantities are calculated and discussed physically. Finally, the proposed model is validated with experimental data through suspended sediment concentration distribution. It may be noted that most of the cases here are considered for one-dimensional (vertical position) configuration, which is more important in the context of vertical suspended sediment concentration distribution. However, extension to the two-dimensional case is straightforward and can be achieved without any difficulty.  
\subsection{Numerical Convergence of RSDPTM}
\label{subsec:numerical_convergence}
We simulate \cref{eq24} in $z$ direction using the projected EM method to test the strong order of convergence numerically. All the parameter values are taken from Run 13 of \cite{coleman1981velocity} data, as can be seen from \cref{table2}. To estimate the strong error, we need to know a reference solution. The reference solution is considered using two ways: one is the projected EM method using a smaller time step, namely $\Delta t=2^{-15} $, and the other is Milstein's method for $\Delta t=2^{-15}$. Milstein's method is a high-resolution numerical method for solving SDEs \citep{kloeden2011numerical}. First, we discretize 10000 Brownian paths using a time step $\delta t = 2^{-15}$, and then solve \cref{eq24} along these paths using projected EM method with time steps $\Delta t = \left\{2^{-5},...,2^{-10}\right\}$. Then, the errors are calculated at the end time $T=1$ sec, considering the time domain $\left[0,1\right]$. To check the order of error, we also plot the reference line with slope 1/2. The slopes of the lines representing the errors are calculated as 0.48 and 0.45 for \cref{fig3a} and \cref{fig3b}, respectively. Thus, it is tested numerically that the projected EM method has a strong order of convergence 1/2.

\begin{table}[ht]
	\scalebox{0.75}{\begin{tabular}{|c|c|c|c|c|c|c|c|}
			\hline
			Data                                                                      & Run & \begin{tabular}[c]{@{}c@{}}Settling velocity\\ $w_{s}$ (m/s)\end{tabular} & \begin{tabular}[c]{@{}c@{}}Flow depth\\ $h$ (m)\end{tabular} & \begin{tabular}[c]{@{}c@{}}Particle diamter\\ $d$ (mm)\end{tabular} & \begin{tabular}[c]{@{}c@{}}Shear velocity\\ $u_{*}$ (m/s)\end{tabular} & \begin{tabular}[c]{@{}c@{}}Reference level\\ $a$ (m)\end{tabular} & \begin{tabular}[c]{@{}c@{}}Schmidt number\\ Sc\end{tabular} \\ \hline
			\multirow{3}{*}{\begin{tabular}[c]{@{}c@{}}Coleman\\ (1986)\end{tabular}} & 3   & 0.007                                                                      & 0.172                                                            & 0.105                                                                   & 0.041                                                                     & $6.020 \times 10^{-3}$                                                               & 0.707                                                           \\ \cline{2-8} 
			& 8   & 0.007                                                                      & 0.173                                                            & 0.105                                                                   & 0.041                                                                      & $6.060 \times 10^{-3}$                                                                 & 0.592                                                           \\ \cline{2-8} 
			& 13   & 0.007                                                                      & 0.171                                                            & 0.105                                                                   & 0.041                                                                      & 5.985 $\times$ $10^{-3}$                                                                 & 0.551                                                           \\ \hline
	\end{tabular}}
	\caption{Experimental conditions of \cite{coleman1981velocity} data.}
	\label{table2}
\end{table}

\begin{figure}[hbt!]
	\centering
	\subfloat[]{\includegraphics[height=7cm,width=8cm]{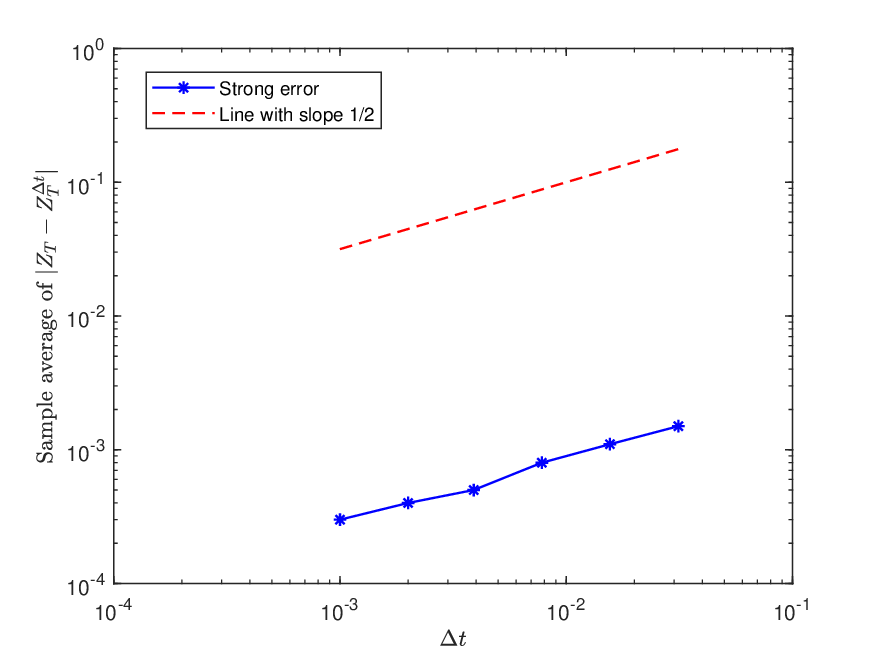}\label{fig3a}}\hfill
	\subfloat[]{\includegraphics[height=7cm,width=8cm]{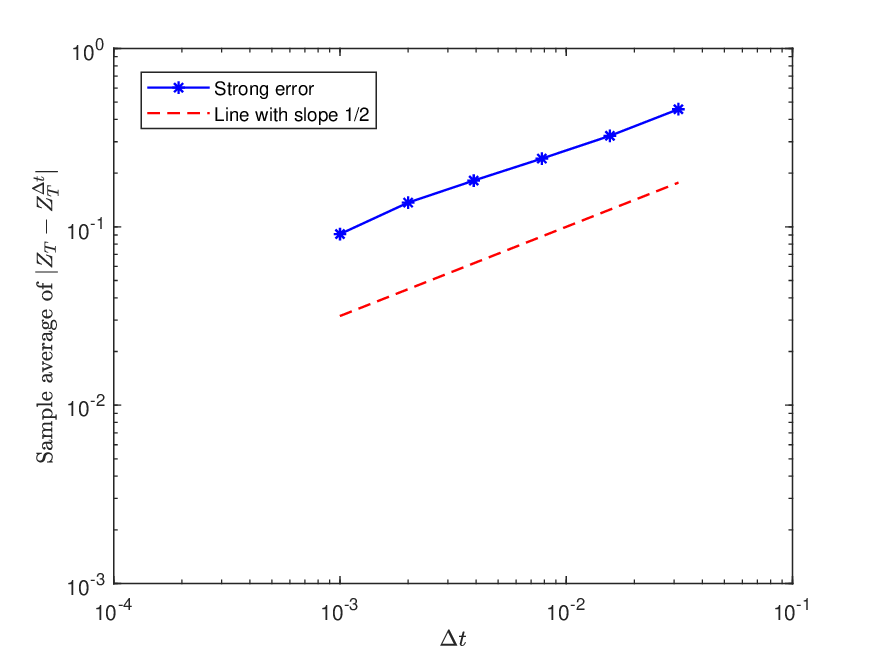}\label{fig3b}}
	\caption{Strong error plots: (a) projected EM with $\Delta t = 2^{-15}$, and (b) Milstein's method with $\Delta t =2^{-15}$ are chosen as the reference solutions.}
	\label{fig3}
\end{figure}

\subsection{Comparison between Euler and Lagrange Frameworks}
\label{subsec:compare euler lagrange}
We derived both the Eulerian and Lagrangian models for the sediment dynamics. The Eulerian approach results in the FPE \cref{eq10} with the boundary conditions \cref{eq15}, and the Lagrangian model derives the RSDPTM \cref{eq24}. Here, we compare the probability density functions (PDFs) obtained from these two approaches considering the one-dimensional case (vertical direction, $z$). All the parameter values are taken from Run 13 of \cite{coleman1981velocity} data and the time interval for the simulation is considered as $\left[0,30\right]$. For RSDPTM, 50000 simulations are performed considering the initial position of the particle at the water surface, and then the histograms are plotted in \cref{fig4a} for different instants of time. On the other hand, FPE \cref{eq10} is solved using the MATLAB toolbox \textit{pdepe}, which uses high-resolution numerical schemes of parabolic PDEs \citep{skeel1990method}. The resulting PDFs are plotted in \cref{fig4b}. It is seen from \cref{fig4} that the two distributions coincide, which indicates that the Eulerian (\cref{eq10}) and Lagrangian (\cref{eq24}) frameworks for developing RSDPTM  are consistent with each other. 

\begin{figure}[hbt!]
	\centering
	\subfloat[]{\includegraphics[height=7cm,width=8.2cm]{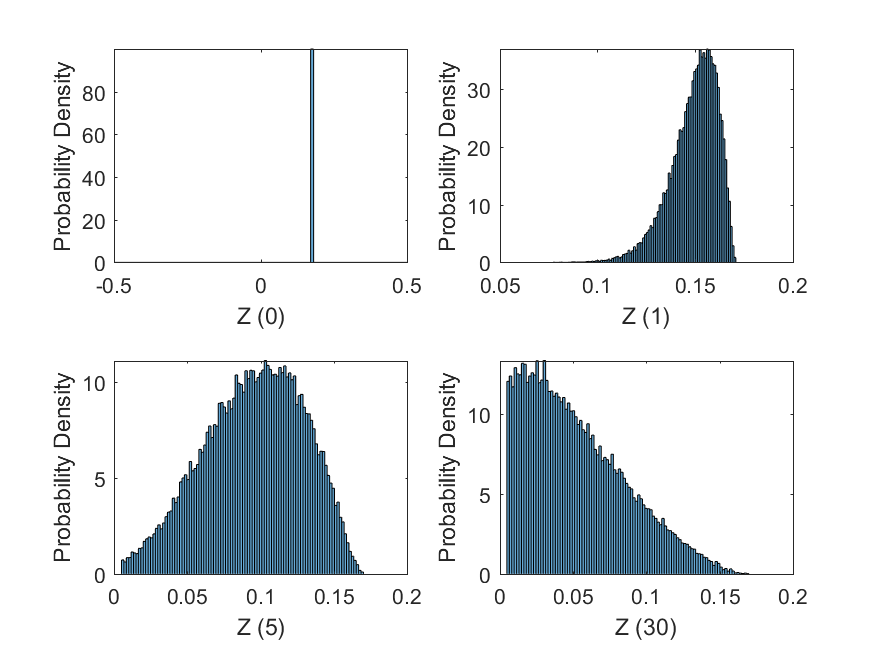}\label{fig4a}}\hfill
	\subfloat[]{\includegraphics[height=7cm,width=8.2cm]{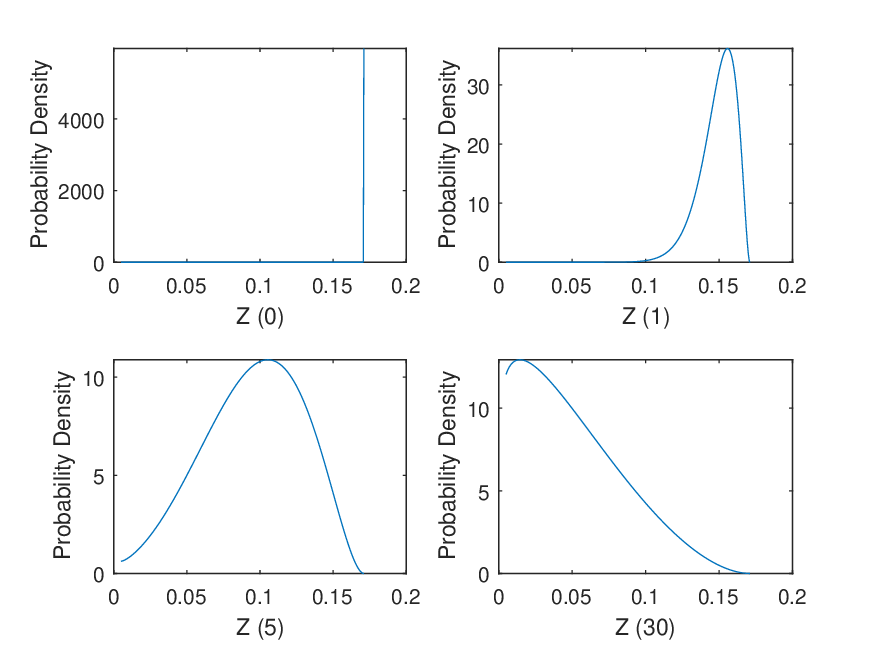}\label{fig4b}}
	\caption{Comparison of PDFs: (a) Lagrange (50K particles), and (b) Euler.}
	\label{fig4}
\end{figure}

\subsection{Sample Trajectories and Ensemble Statistics}
\label{subsec:compare euler lagrange}
The particles can be placed inside the domain whenever they exceed the boundary in several ways. While we used a mathematically proven reflected process, it can also be tackled heuristically. Therefore, it may be interesting to compare the behaviour of the particles' movements using different approaches. For that reason, when a particle exceeds the boundary, apart from the reflected process developed in this work (RSDPTM), we also consider two other cases of tackling the situation, namely (i) placing the particle at the boundary (case I), and (ii) set the current position as previous position, which is inside the domain (case II). The RSDPTM (\cref{eq24}), cases I and II are simulated using 50000 particles, time step $\Delta t=0.01$ sec, and the initial position $(X_{0},Z_{0}) = (0,h)$. Some of the sample trajectories are plotted in \cref{fig5} for each of these cases. It can be observed from the figure that all the techniques constrain the particles to move within the given bounded domain. However, their characteristics are different due to their underlying techniques. Considering the same simulation condition, the ensemble means, and variances of particle trajectories in the streamwise $(x)$ and vertical $(z)$ directions are plotted in \cref{fig6,fig7}, respectively. The ensemble mean and variance in $x$ direction increases with time, as can be seen from \cref{fig6a,fig7a}. In the vertical direction, the particle starts from the water surface $(z=h)$, and the mean gradually decreases over time $t$, which is observed from \cref{fig6b}. In \cref{fig7b}, the variance initially increases with $t$ and then reaches a stable value after some period of time. This happens because the particles tend to settle on the bed due to gravitation settling. Further, \cref{fig6} shows that different approaches, namely RSDPTM, cases I and III, have negligible effect in the ensemble mean of particle trajectories. This behavior suggests that, on average, the particle trajectories are similar. This could be due to the fact that, over many realizations, the particles tend to behave similarly in terms of their average movement. On the other hand, in \cref{fig7}, significant differences are observed for the case of ensemble variance. The differences in variance indicate that the different boundary-handling approaches affect the spread or dispersion of particle trajectories. It may be noted that the orthogonal projection mechanism representing RSDPTM is useful in modeling the boundary-constrained stochastic process; however, it is required to be modified in order to incorporate specific physical characteristics of the flow behavior. To that end, we propose an improved algorithm for sediment particle trajectories in an open channel turbulent flow. 
\begin{figure}[hbt!]
	\centering
	\subfloat[]{\includegraphics[height=4cm,width=5.5cm]{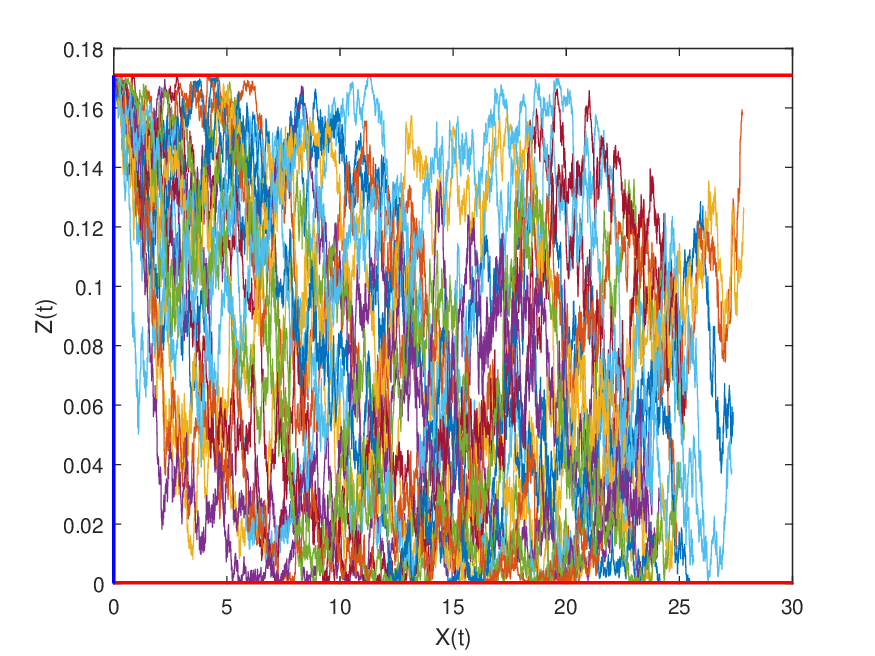}\label{fig5a}}\hfill
	\subfloat[]{\includegraphics[height=4cm,width=5.5cm]{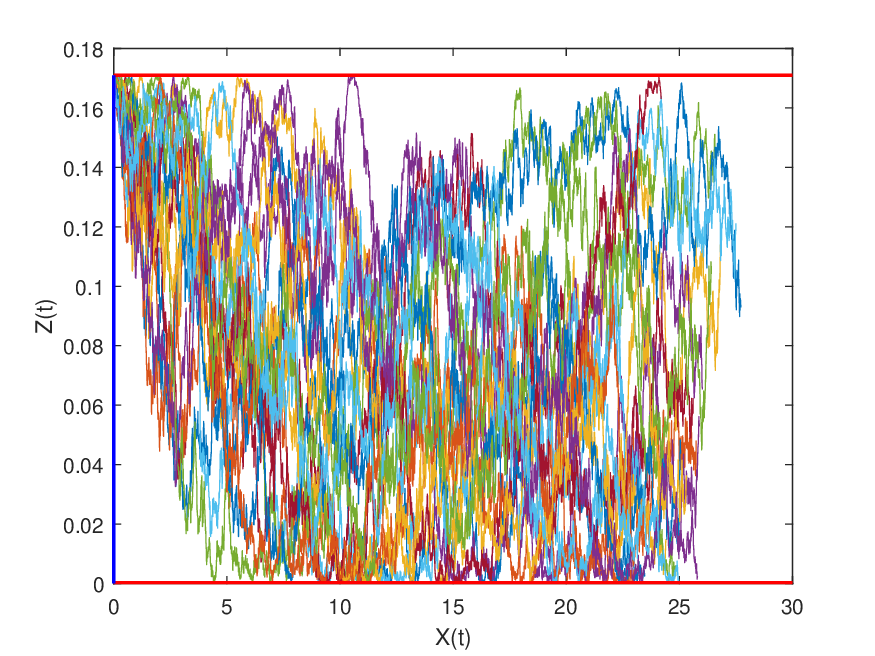}\label{fig5b}}\hfill
	\subfloat[]{\includegraphics[height=4cm,width=5.5cm]{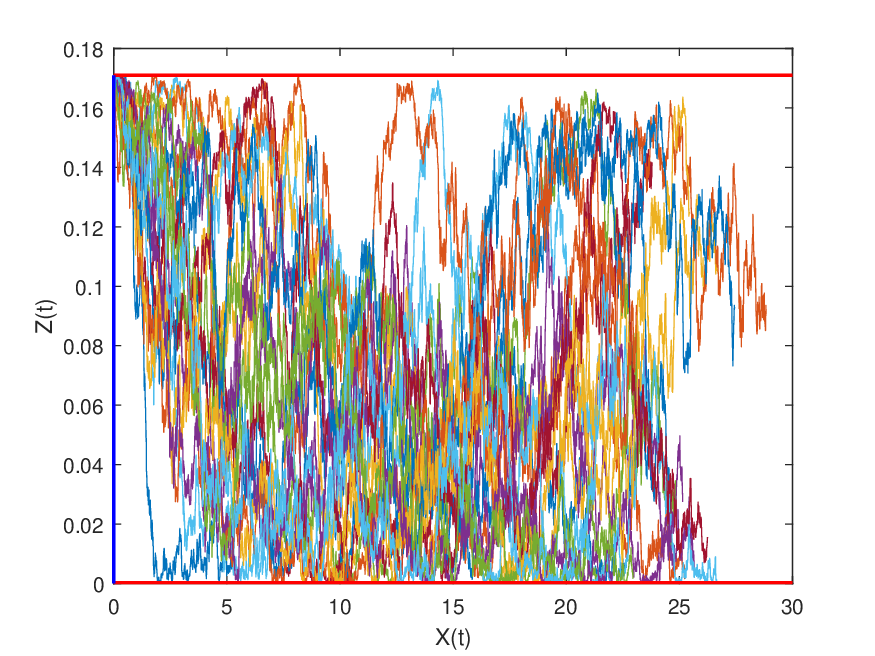}\label{fig5c}}
	\caption{Selected sample trajectories from 50000 particle simulations: (a) RSDPTM, (b) SDPTM (case I), and (c) SDPTM (case II).}
	\label{fig5}
\end{figure}

\begin{figure}[hbt!]
	\centering
	\subfloat[]{\includegraphics[height=6.5cm,width=8.2cm]{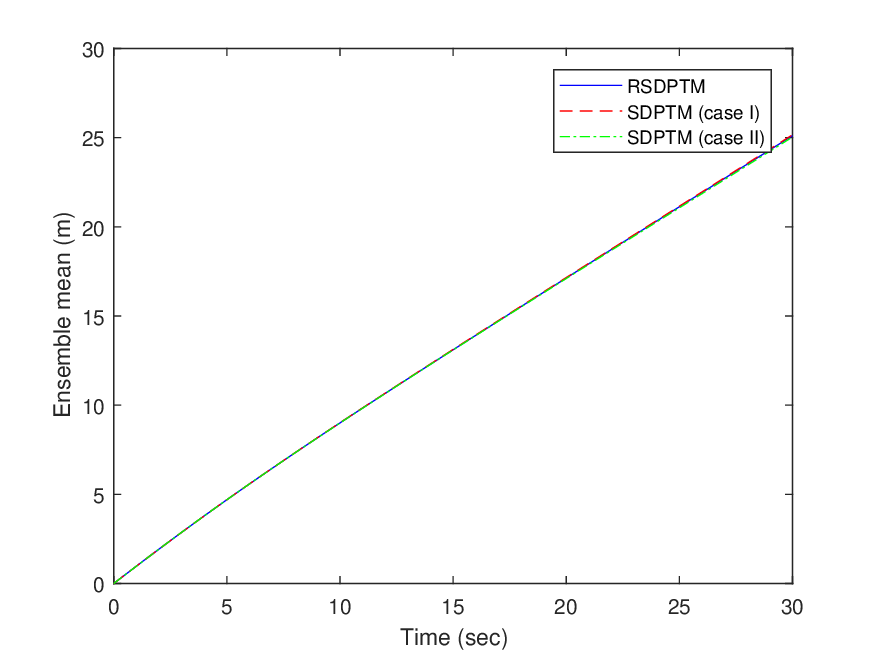}\label{fig6a}}\hfill
	\subfloat[]{\includegraphics[height=6.5cm,width=8.2cm]{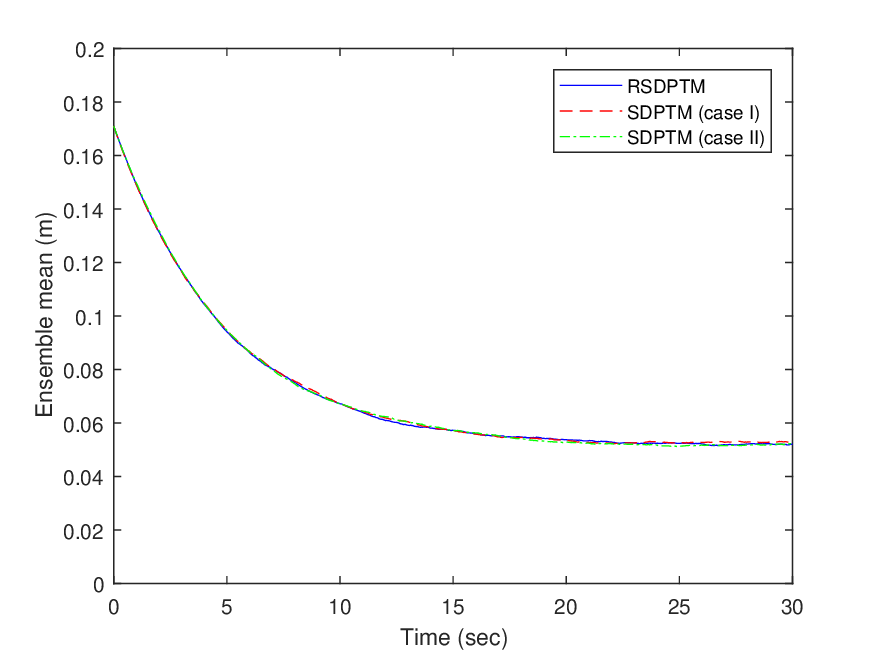}\label{fig6b}}
	\caption{Ensemble means of particle trajectories: (a) streamwise ($x$), and (b) vertical $(z)$ direction.}
	\label{fig6}
\end{figure}

\begin{figure}[hbt!]
	\centering
	\subfloat[]{\includegraphics[height=6.5cm,width=8.2cm]{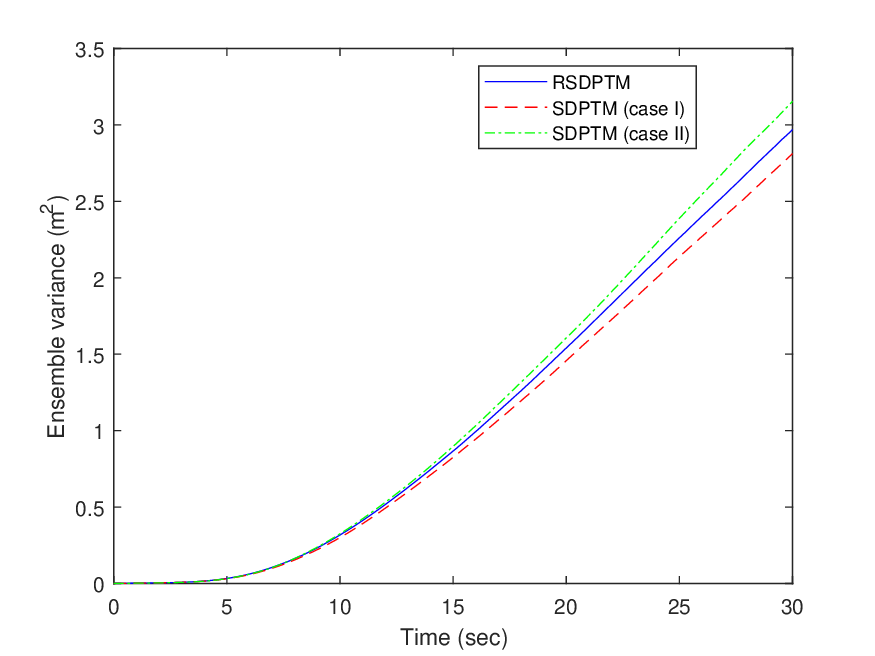}\label{fig7a}}\hfill
	\subfloat[]{\includegraphics[height=6.5cm,width=8.2cm]{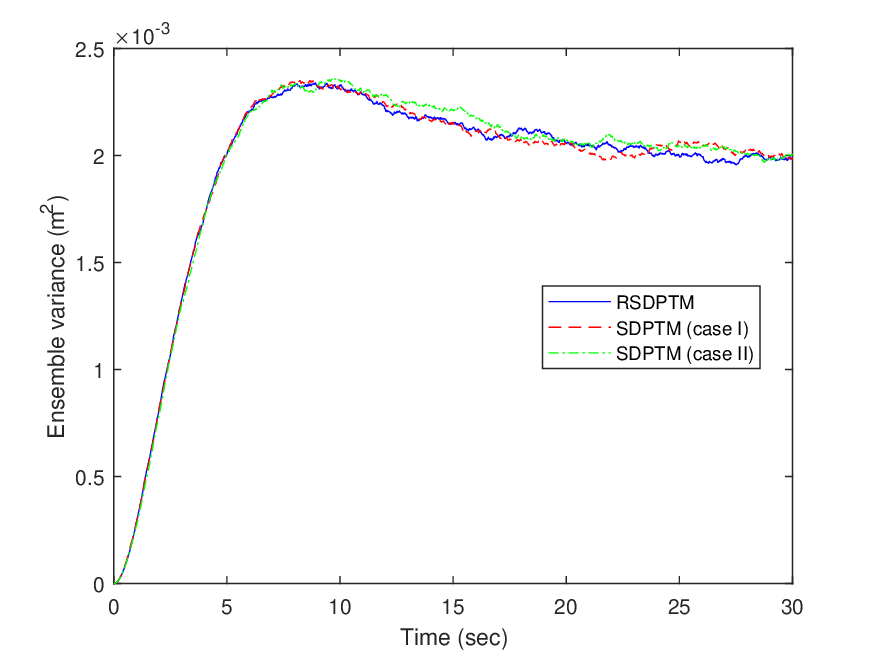}\label{fig7b}}
	\caption{Ensemble variances of particle trajectories: (a) streamwise ($x$), and (b) vertical $(z)$ direction.}
	\label{fig7}
\end{figure}

\subsection{Improved Algorithm Incorporating Resuspension Mechanism}
\label{subsec:improved algorthim}
Sediment particles forming the bed are set in motion under certain threshold conditions. The particles near the bed, which are coarser in nature, move by sliding, rolling, or performing small jumps (saltation). On the other hand, the finer particles are lifted up due to the production of turbulence and its upward diffusion to remain in suspension over an appreciable period of time. In reality, the particles intermittently come into contact with the bed and are shifted through more or less jumps to remain surrounded by the fluid. The threshold of sediment suspension is defined as the flow condition required for the initiation of suspension. Long back, \cite{bagnold1966approach} developed the threshold condition in terms of settling velocity $w_{s}$ and shear velocity $u_{*}$. Following this concept, several researchers formulated the suspension threshold using $w_{s}$ and $u_{*}$ \citep{xie1981river,rijn1984sediment,sumer1986recent,celik1991suspended}. However, these studies are empirical-based and do not consider probabilistic aspects. \cite{cheng1999analysis}  
employed the probabilistic concept for the first time to determine the threshold condition for sediment suspension from bed load. They stated that when the vertical velocity fluctuation exceeds the downward settling velocity, sediment particles come in suspension and define the probability of sediment suspension as $p(w' > w_{s})$. They assumed the velocity fluctuation to follow a Gaussian distribution. Later, \cite{bose2013sediment} suggested that the Gaussian distribution results from adding errors, but turbulent velocity fluctuations do not follow this pattern. They proposed the Gram-Charlier series expansion of the probability densities based on the two-sided exponential or Laplace distribution. The probability distribution was estimated as:
\begin{equation}\label{eq39}
	p\left(w'\right) = \begin{cases}
		\frac{1}{\sqrt{\overline{w'^{2}}}} \left(17+ \hat{w} - \hat{w}^{2}\right) \exp \left(-\hat{w}\right) & \text{if } w' \geq 0 \\
		0 & \text{if } w' < 0 
	\end{cases}
\end{equation}
where $\hat{w} = w' / \sqrt{\overline{w'^{2}}}$, and $\sqrt{\overline{w'^{2}}}$ is the root-mean-square (rms) of $w'$. The rms of $w'$ for hydraulically rough and smooth flow regime can be given as \citep{dey2014fluvial}:
\begin{equation}\label{eq39_rms}
	\sqrt{\overline{w'^{2}}} = \begin{cases}
		u_{*} & \text{for hydraulically rough flow regime}\\
		u_{*}\left(1-\exp\left[-0.025\left(\frac{2.75u_{*}d}{\nu_{f}}\right)^{1.3}\right]\right) & \text{for hydraulically smooth flow regime}
 	\end{cases}
\end{equation}
\cite{bose2013sediment} used \cref{eq39} and the suspension criteria $w' > w_{s}$ to deduce the probability $p(w' > w_{s})$. Following the above discussion, it is important to consider the suspension mechanism in our model to formulate the suspended sediment dynamics stochastic model correctly. Further, the orthogonal projection technique resulting in RSDPTM \cref{eq24} is mostly a mathematical tool and should be supplemented by some additional steps in simulating sediment particle trajectories. To that end, we propose an improved algorithm, where the steps are given in what follows. 
\par
Choosing a proper time step $\Delta t$ and initial values $X_{0}$, $Z_{0}$, first, we simulate the particle trajectories in streamwise ($x$) and vertical ($z$) directions using SDPTM \cref{eq35}. 
\newline
\newline
\textbf{Step 1.} $$ 
X'_{t+\Delta t} = X_{t} + \left(\frac{u_{*}}{\kappa}\ln \frac{Z_{t}}{z_{0}}\right) \Delta t + \sqrt{0.30 u_{*} Z_{t}} dB_{1t}$$
$$Z'_{t+\Delta t} = Z_{t} + \left(-w_{s}+\kappa u_{*} \Sc \left(1-\frac{2Z_{t}}{h}\right)\right) \Delta t + \sqrt{2 \kappa u_{*} \Sc Z_{t}\left(1-\frac{Z_{t}}{h}\right)} dB_{2t} $$
Then, we check whether the particle location remains within or outside the domain, $D=[0,\infty) \times [a,h]$. If it goes beyond the boundary, then the orthogonal projection is applied.\newline
\newline
\textbf{Step 2.} At the same time-step $t+\Delta t$, \begin{equation}\label{eq40}
	X'' = \begin{cases}
-X' & \text{if } X' < 0 \\
		X' & \text{if } X' \geq 0 
	\end{cases},~~~~~~~~ Z'' = \begin{cases}
	2a-Z' & \text{if } Z' < a \\
	2h-Z' & \text{if } Z' > h\\
	Z' & \text{otherwise} 
	\end{cases}
\end{equation}
According to the above equation, the particles are returned to the domain once they exceed the boundary. This follows from a mathematically consistent theory of reflected process. However, regarding the physical mechanism, it should be noted that the particles are exposed to the flow when they enter the domain. Specifically, they are carried away by mean drift velocity in the vertical direction and by the main flow velocity along the streamwise direction. Hence, the following step is added at the same time-step. \newline
\newline
\textbf{Step 3.} At the same time-step $t+\Delta t$, 
\begin{equation}\label{eq41}
	X''' =  X'' + \left(\frac{u_{*}}{\kappa}\ln \frac{Z_{t}}{z_{0}}\right) \Delta t,~~~~Z'''=Z'' + \left(-w_{s}+\kappa u_{*} \Sc \left(1-\frac{2Z_{t}}{h}\right)\right) \Delta t
\end{equation}
The particles in Step 4 may or may not reach the channel bed from where the suspension happens. We place them on the bottom boundary in case they exceed the boundary again. \newline
\newline
\textbf{Step 4.} 
\begin{equation}\label{eq42}
	Z_{t+\Delta t} = \begin{cases}
		a & \text{if } Z'''\leq a\\
		Z''' & \text{otherwise}
	\end{cases}
\end{equation}  
As discussed previously, the particles on the bed are subjected to resuspension mechanism. Therefore, in Step 4, we check the suspension criteria in accordance with \cref{eq39}, and if it is satisfied, then we assign a new elevation length for the re-suspended particle. This entertainment elevation should be taken as a random value subject to specific flow and sediment condition \citep{macdonald2006ptm}. In Eulerian modeling, the Rouse equation is the most widely used suspended sediment concentration distribution profile. This equation can be written as:
\begin{equation}\label{rouse-type}
\frac{C}{C_{a}} = \left(\frac{h-z}{z} \frac{a}{h-a}\right)^{\frac{w_{s}}{\kappa Sc u_{*}}}
\end{equation} 
where $C_{a}$ is the sediment concentration at the reference level $z=a$. We generate Rouse-type random number, say, $R_{\text{Rouse}}$, which follows the characteristics of \cref{rouse-type}. Specifically, this method will generate random numbers that are more biased towards $z=a$ (where sediment concentration is high) and less biased towards $z=h$ (where sediment concentration is low).  \newline
\newline
\textbf{Step 5.} When $Z''' \leq a$ in Step 4 \begin{equation}\label{eq43}
	Z_{t+\Delta t} = \begin{cases}
		R_{\text{Rouse}} & \text{if } w' > w_{s}\\
		a & \text{otherwise}
	\end{cases}
\end{equation}
Steps 1-5 constitute an improved algorithm for tracking sediment particles using RSDPTM. 
\par 
The statistical characterization of the diffusion phenomenon is quantified using the mean-square-displacement (MSD). It can be represented as follows:
\begin{equation}\label{eq44}
	\sigma_{s}^{2} = \left<\left(s_{i}-\left<s\right>\right)^{2}\right> \propto t^{2\gamma}
\end{equation}
where $s$ is the space coordinate, $\left<s\right>$ is the mean of particle locations $s_{i}$, $t$ is time, and $\gamma$ is the scaling diffusion exponent. Classical theories regard the movements of sediment particles to be based on Fick's law, which results in normal or Fickian diffusion. In this case, MSD varies linearly with time, i.e., $\gamma=0.5$. However, recent experimental investigations have reported the case of departure from normal diffusion resulting in $\gamma \neq 0.5$. More specifically, $\gamma < 0.5$ represents the sub-diffusion, $\gamma > 0.5$ denotes the super-diffusion, and $\gamma=2$ indicates the ballistic motion. Anomalous diffusion has been examined for both bed-load particles \citep{bradley2010fractional,fan2016exploring,martin2012physical,nikora2002bed,saletti2015temporal} and suspended particles \citep{afonso2014anomalous,chen2016superdiffusion,nie2017vertical,park2018modeling,schumer2009fractional,tsai2020stochastic}. We consider 50000 particles and plot the MSDs in \cref{fig8} for both $x$ and $z$ directions. The values of the scaling diffusion exponents in specific time periods are reported in \cref{table1}. We aim to examine the movement of sediment particles being deviated from normal or Ficikian diffusion. \cref{fig8a} shows that the diffusion mechanism in streamwise direction follows ``sub-diffusion $\to$ super-diffusion $\to$ super-diffusion (ballistic motion) $\to$ super-diffusion''. On the other hand, the vertical diffusion mechanism is found to follow the sequence ``super-diffusion $\to$ sub-diffusion $\to$ sub-diffusion ($\gamma_{z} \approx 0$)''. It can be observed from \cref{table1} that at the beginning ($t \leq 1$ sec), the particle motions in streamwise and vertical directions are slightly sub- and super-diffusive, respectively. As time increases, particle motion in both directions exhibits non-Fickian or anomalous diffusion. Specifically, in the $x$ direction, particle motion changes from sub- to super-diffusion with increasing intensity. This can attributed to the accelerated and decelerated nature of the motion caused by turbulent eddies, leading to enhanced diffusion. However, in the $z$ direction, the motion changes from super- to sub-diffusion with decreasing intensity. This observation can be ascribed to the movements of particles, encompassing processes such as deposition and resuspension. When sediment particles are introduced at the water surface, they tend to settle towards lower regions in the flow under the influence of gravity. Consequently, the variances in particle positions along the vertical axis decrease as simulation time progresses, leading to subdiffusion. This is attributed to spatial constraints, specifically the boundary effect. After a certain duration, some particles settle on the bed while others continue to be transported. In \cref{fig9}, we plot sediment clouds for 500 realizations at different times, namely $t=5, 10, 20, 30,$ and 50 secs. Initially, the particles are released at the water surface, and then they start to settle. The particle clouds are concentrated initially and become more dispersed as time increases. 
\begin{figure}[hbt!]
	\centering
	\subfloat[]{\includegraphics[height=6.5cm,width=8.2cm]{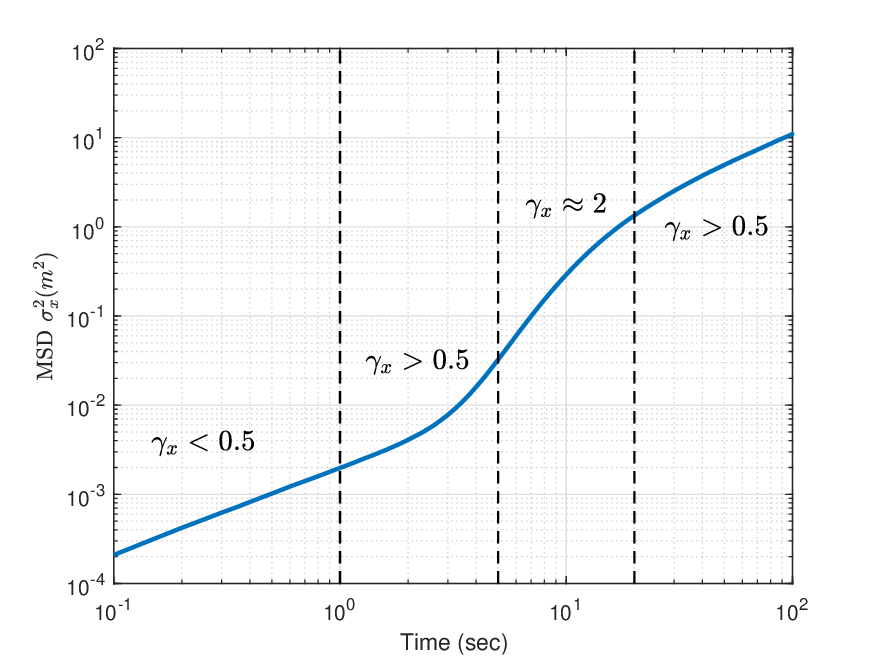}\label{fig8a}}\hfill
	\subfloat[]{\includegraphics[height=6.5cm,width=8.2cm]{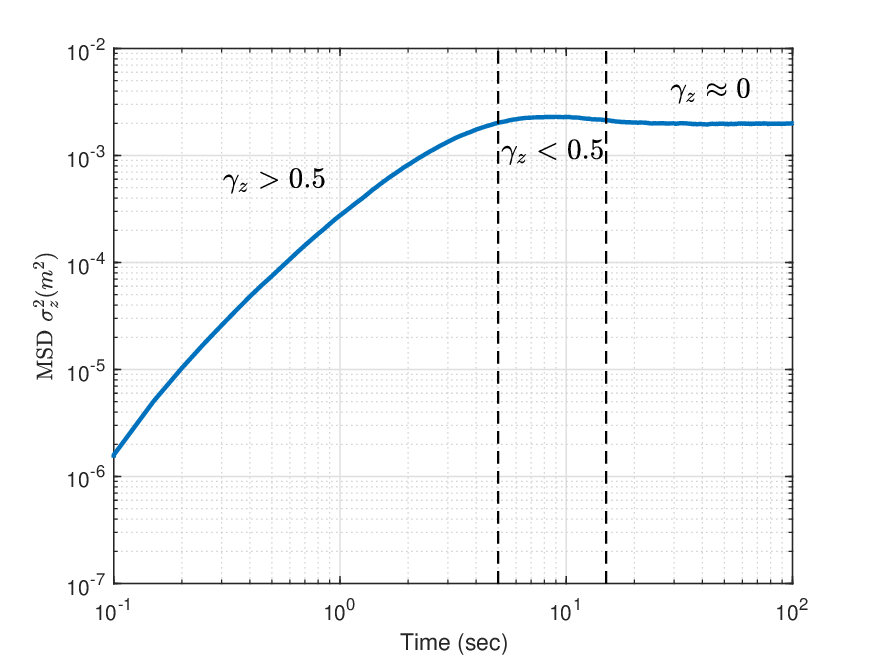}\label{fig8b}}
	\caption{MSD of particle trajectories: (a) streamwise ($x$), and (b) vertical $(z)$ direction.}
	\label{fig8}
\end{figure}

\begin{table}[ht]
	\centering
	\scalebox{0.92}{\begin{tabular*}{\textwidth}{@{\extracolsep{\fill}}ccc}
		\toprule
		\textbf{Time periods (sec)} & \textbf{Exponent $\gamma_{x}$} & \textbf{Exponent $\gamma_{z}$}\\
		\midrule
		0-1 & 0.48 & 0.52\\
		\midrule
		1-5 & 0.66 & 0.78\\
		\midrule
		5-10 & 2.03 & 0.42\\
		\midrule
		10-15 & 2.05 & 0.02\\
		\midrule
		15-50 & 1.14 & 0.00\\
		\midrule
		50-100 & 0.55 & 0.00\\
		\bottomrule
	\end{tabular*}}
	\caption{Scaling diffusion exponents in specific time periods.}
	\label{table1}
\end{table}

\begin{figure}[hbt!]
	\centering
	\includegraphics[height=8cm,width=15cm]{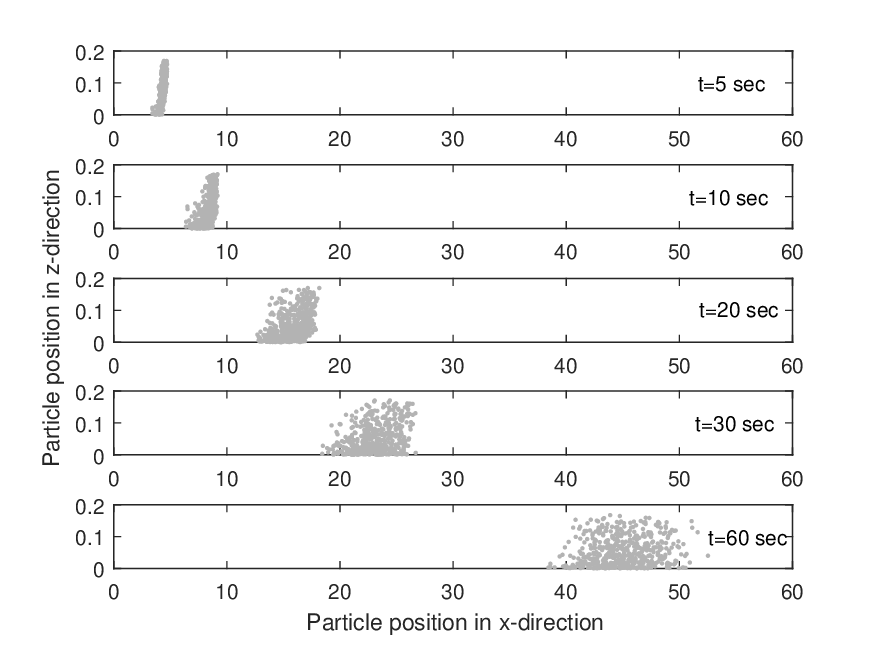}
	\caption{Sediment particle cloud in the flow.}
	\label{fig9}
\end{figure}

\subsection{Model Validation through Experimental Data}
\label{subsec:data validation}
To check the efficiency of the proposed algorithm based on RSDPTM in the previous section, we validate it through the experimental data of suspended sediment concentration (SSC) distribution. The dilute sediment-laden flow data of \cite{coleman1981velocity} are chosen for that purpose. In his study, \cite{coleman1981velocity} utilized a smooth flume with dimensions of 356 mm in width and 15 m in length. The flow depth was consistently maintained at approximately 17.1 cm. Out of the 40 test cases conducted, test cases 1, 21, and 32 were executed under clear-water flow conditions. Conversely, test cases 2–20, 22-31, and 33-40 involved a sediment bed with three distinct sand diameters: $d =$ 0.105 mm, $d =$ 0.21 mm, and $d =$ 0.42 mm.
\par 
The vertical distribution of SSC from the particle trajectories given by the RSDPTM model \cref{eq24} is calculated as follows. First, we divide the flow depths into several bins and then calculate the number of particles in each of the bins. Finally, the SSC is estimated by dividing the number of particles in a bin by the total number of particles (simulations). We choose Runs 3, 8, and 13 from \cite{coleman1981velocity} data to validate the model and report the required parameter values in \cref{table2}. For all the cases, 100K particles, time step $\Delta t=0.01$ sec, and 100 bins are chosen. The normalized SSC vs depth is plotted in \cref{fig10}. Initially, in \cref{fig10a}, we choose three different times, namely $t=$ 15, 20, and 60 sec to compare the SSC model with data. It can be seen that as time increases, the model matches better with the data. This is because the experimental data of \cite{coleman1981velocity} were collected for a steady-uniform flow. Thus, the proposed model performs better when it reaches a sufficiently large time and becomes equilibrium. Overall, an excellent agreement is found between the estimated and observed values of data, as can be seen from \cref{fig10}.  

\begin{figure}[hbt!]
	\centering
	\subfloat[Run 03]{\includegraphics[height=6.5cm,width=8.2cm]{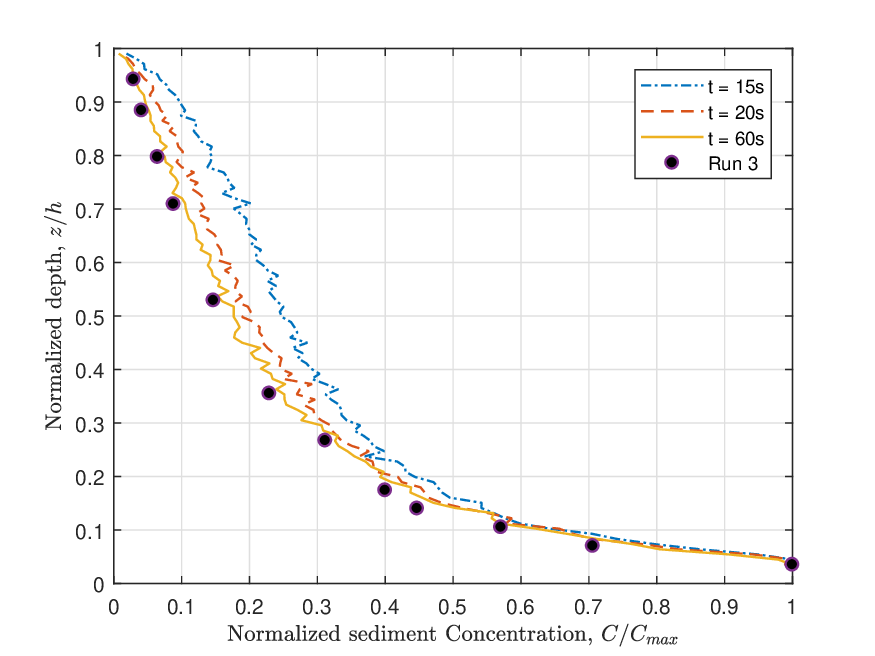}\label{fig10a}}
	\hfill
	\subfloat[Run 08]{\includegraphics[height=6.5cm,width=8.2cm]{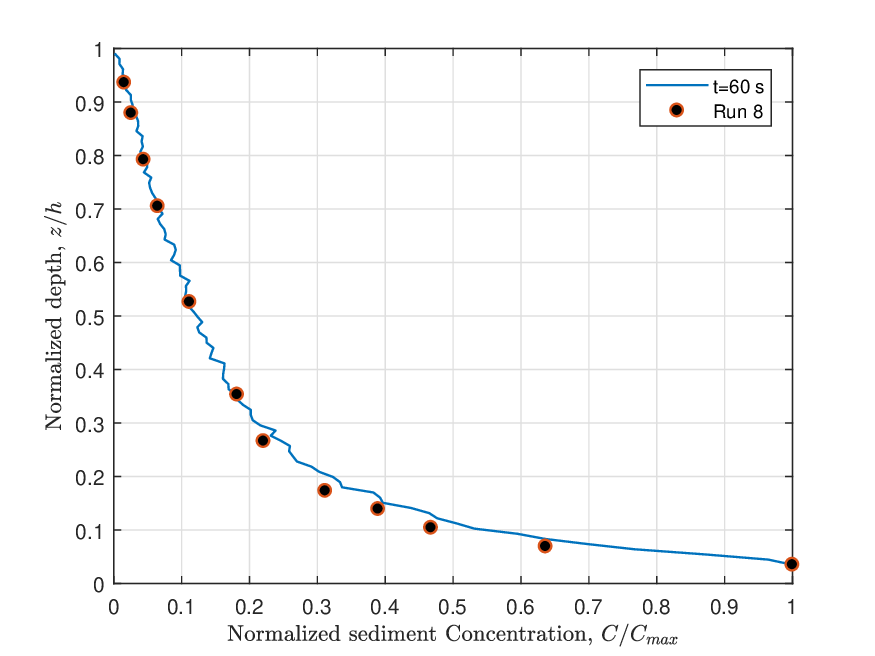}\label{fig10b}}
	
	\vspace{-5pt} 
	
	\subfloat[Run 13]{\includegraphics[height=6.5cm,width=8.2cm]{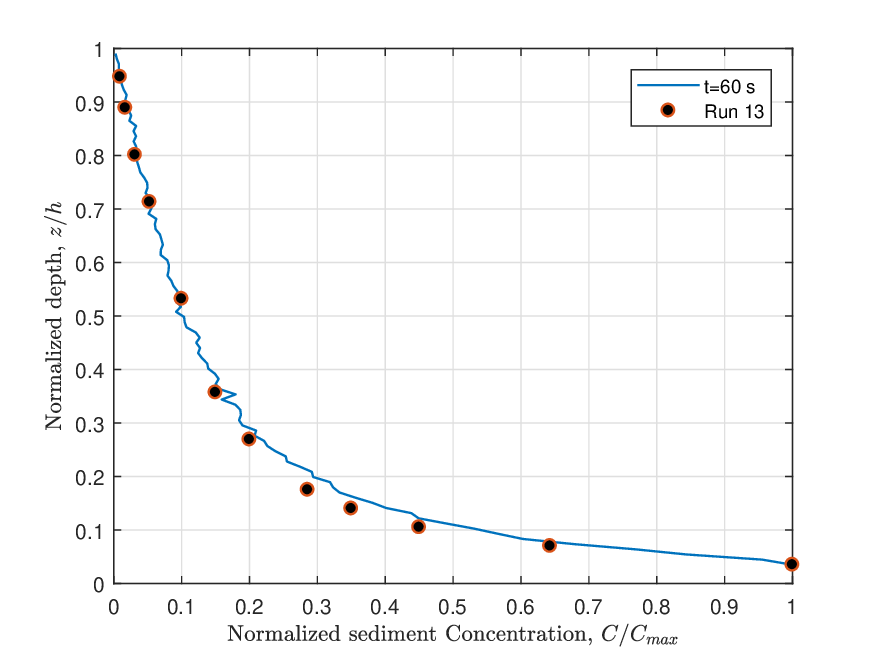}\label{fig10c}}
	
	\caption{Comparison of RSDPTM-based sediment concentration model with (a) Run 3, (b) Run 8, and (c) Run 13 of \cite{coleman1981velocity} data.}
	\label{fig10}
\end{figure}


\section{Conclusions and Future Recommendations}
\label{sec:conclusions}
This work proposes both stochastic Lagrangian and Eulerian models for suspended sediment dynamics in a two-dimensional open-channel turbulent flow. An open channel is bounded vertically by the reference level at the bottom and the water surface at the top and is semi-bounded in the streamwise direction. Therefore, unlike previous studies, the Lagrangian framework derives the reflected SDE incorporating the boundary effects of the flow domain. On the other hand, the Eulerian approach proposes FPE along with consistent boundary conditions. The SDE is simulated numerically using a projected EM method, while the FPE is solved using a Matlab toolbox. 
\par 
Unlike standard SDE, the reflected SDE (RSDE) requires some additional conditions for the existence and uniqueness of the solution. Considering a general mean drift and diffusion term, we prove the existence and uniqueness of the result. This indicates that the SDPTM or other similar stochastic sediment particle motion-based Lagrangian model always has a solution, and it is unique. In general, the numerical solution of RSDEs needs an improved version of the EM method, namely the projected EM method, which operates an orthogonal projection operator to the particles exceeding the boundary of the domain. Thus, the strong/weak order of convergence of the numerical method does not work like standard EM. For the proposed model, the strong order of convergence is checked by considering some numerical tests, and it is seen that the projected EM has a 1/2 strong order of convergence. The comparison between Eulerian (FPE) and Lagrangian (RSDPTM) models is given, which shows their consistency with each other. 
\par 
While the reflected process is a mathematically consistent formulation for a stochastic process in a bounded region, it should be supplemented by some additional steps in order to represent the physical mechanism of suspended sediment dynamics. To that end, the concept of threshold of sediment suspension is introduced, in which the particles may be subjected to resuspension under certain conditions once they reach the channel bed. Incorporating this and some additional steps, an improved algorithm for stochastic sediment particle trajectories is proposed. The mean-squared-displacement (MSD) of particle trajectories in both directions reveal the anomalous diffusion nature of particle motions. Specifically, the motion changes from sub- to super-diffusion and super- to sub-diffusion in streamwise and vertical directions, respectively. Initially, the particles are released at the water surface, and then they start to settle as time increases, which explains a constant variance ($\gamma_{z} \approx 0$) in the vertical direction. Finally, the proposed algorithm is checked for its robustness in application to experimental data validation. The vertical distribution of suspended sediment concentration is calculated from the particle location and then compared with relevant experimental data. An excellent agreement is found between the proposed model and experimental data, which may justify the applicability of the proposed algorithm incorporating RSDPTM. 
\par 
One of the hypotheses of this work is that the sediment particles are light enough to follow the nature of the fluid particles. Therefore, for modeling trajectories of relatively heavier particles, the particle Langevin alone may not be a perfect approach. Indeed, the flow field should be treated as a two-phase medium, and the Lagrangian approach for both position and velocity should be adopted \citep{minier2014guidelines}. Particularly, the velocity field is driven by a specific stochastic process, namely the Ornstein-Uhlenbeck (OU) process \citep{viggiano2020modelling}. Nevertheless, the proposed RSDPTM for handling the confined domain can still be a potential candidate for correctly formulating the stochastic model. Further, recent developments have shown that the turbulence-coherent structures contribute notably to the sediment suspension and transportation mechanism. Turbulent bursting events, specifically ejection and sweep events, are responsible for the suspension and deposition of sediment particles. In future studies, we aim to employ the concept of the bursting phenomenon in the RSDPTM by analyzing the Direct Numerical Simulation (DNS) datasets of turbulence.  

\section*{Acknowledgments} 
The authors acknowledge the National Science and Technology Council of Taiwan for financially supporting this research under Contract Nos. NSTC 111-2221-E-002-057-MY3. The first author thanks Mr. Wei-Min Shen (Research Assistant, Department of Civil Engineering, National Taiwan University, TAIWAN), Dr. Serena Y Hung (Former PhD Student, Department of Civil Engineering, National Taiwan University, TAIWAN) and Dr. Falguni Roy (Assistant Professor, Department of Mathematics, NIT Karnataka Surathkal, INDIA) for their helpful discussions. 
\bibliographystyle{elsarticle-harv} 
\bibliography{cas-refs}





\end{document}